%% file: main.tex
\pdfoutput=1
\documentclass[acmsmall]{acmart}

\acmJournal{PACMPL}
\acmVolume{1}
\acmNumber{CONF} 
\acmArticle{1}
\acmYear{2018}
\acmMonth{1}
\acmDOI{} 
\startPage{1}

\setcopyright{none}

\usepackage{booktabs}   
\usepackage{subcaption} 

\usepackage{amsmath,stmaryrd}
\usepackage{graphicx}
\usepackage{color}
\usepackage{hyperref} 
\usepackage[normalem]{ulem}

\usepackage{algorithm}
\usepackage{algpseudocode}

\theoremstyle{definition}
\newtheorem{principle}{Principle}{\bfseries}{\itshape}
\theoremstyle{remark}
\newtheorem{remark}{Remark}

\usepackage{proof}
\usepackage{bcprules}

\newif\ifdraft\draftfalse

\newif\ifwithappendix\withappendixtrue

\newif\ifverylong\verylongfalse

\ifdraft
\newcommand{\tk}[1]{\textcolor{red}{[{#1}-Tsukada]}}
\newcommand{\hu}[1]{\textcolor{blue}{[{#1}-Unno]}}
\newcommand{\tkchanged}[1]{\textcolor{red}{#1}}
\newcommand{\huchanged}[1]{\textcolor{blue}{#1}}
\newcommand\removed[1]{{\color{red}\xout{#1}}}
\else
\newcommand{\tk}[1]{}
\newcommand{\hu}[1]{}
\newcommand{\tkchanged}[1]{#1}
\newcommand{\huchanged}[1]{#1}
\newcommand\removed[1]{}
\fi

\newcommand\impact{{\sc Impact}}
\newcommand\spacer{{\sc Spacer}}
\newcommand\recmc{{\sc RecMC}}
\newcommand\cpa{{\sc CPAchecker}}

\newcommand{\IndPDRmbp}{\textsc{IndPDR/mbp}}
\newcommand{\IndPDR}{\textsc{IndPDR}}
\newcommand{\IndMCR}{\textsc{IndMCR}}
\newcommand{\NaiveMCR}{\textsc{Na\"iveMCR}}

\newcommand{\MCImpact}{\textsc{Impact/mc}}

\begin{document}

\title{Software Model-Checking as Cyclic-Proof Search}


\author{Takeshi Tsukada}

\orcid{0000-0002-2824-8708}             
\affiliation{
  \institution{Chiba University}            
  \country{Japan}                    
}
\email{tsukada@math.s.chiba-u.ac.jp}          

\author{Hiroshi Unno}

\orcid{0000-0002-4225-8195}             
\affiliation{
  \institution{University of Tsukuba}           
  \country{Japan}
}
\affiliation{
  \institution{RIKEN AIP}           
  \country{Japan}                   
}
\email{uhiro@cs.tsukuba.ac.jp}         

\input{abst.tex}

\ccsdesc[500]{Software and its engineering~General programming languages}
\ccsdesc[300]{Social and professional topics~History of programming languages}


\maketitle

\input{intro}
\input{background}
\input{basic}

\input{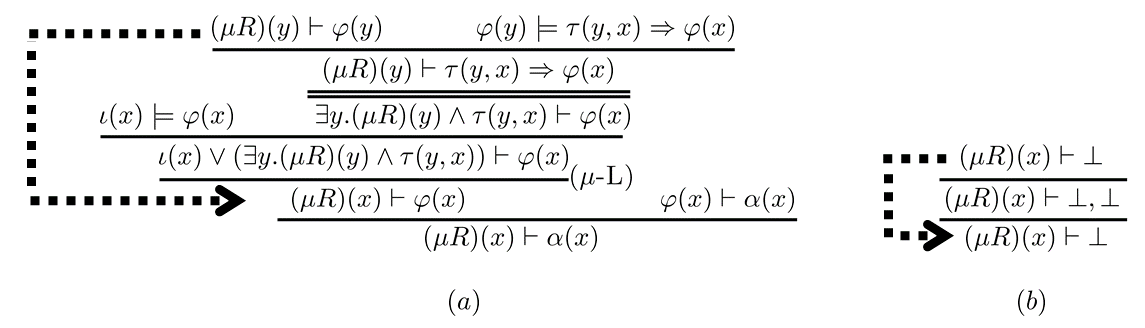}
\input{pdr}

\input{game}
\input{related}
\input{conc}
\begin{acks}                            
  This work was supported by \grantsponsor{JST}{JST}{} ERATO HASUO Metamathematics for Systems Design Project (No. \grantnum{JST}{JPMJER1603}) and \grantsponsor{JSPS}{JSPS}{} KAKENHI Grant Numbers \grantnum{JSPS}{JP20H05703}, \grantnum{JSPS}{JP19K22842}, \grantnum{JSPS}{JP20H04162}, \grantnum{JSPS}{JP17H01720}, and \grantnum{JSPS}{JP19H04084}.
\end{acks}

\bibliography{prog_lang}

\ifwithappendix
\clearpage
\appendix
\input{cyclic-proof-system}
\input{appx-pdr}
\input{appx-diverge-spacer}
\ifverylong
\clearpage
\input{tree}
\fi
\fi

\end{document}

%% file: abst.tex
\begin{abstract}
    %
    This paper shows that a variety of software model-checking algorithms can be seen as proof-search strategies for a non-standard proof system, known as a \emph{cyclic proof system}.
    %
    %
    %
    Our use of the cyclic proof system as a logical foundation of software model checking enables us to compare different algorithms, to reconstruct well-known algorithms from a few simple principles, and to obtain soundness proofs of algorithms for free.
    Among others, we show the significance of a heuristics based on a notion that we call \emph{maximal conservativity}; this explains the cores of important algorithms such as property-directed reachability (PDR) and reveals a surprising connection to an efficient solver of games over infinite graphs that was not regarded as a kind of PDR.
    %
\end{abstract}


%% file: intro.tex
\section{Introduction}
\label{sec:intro}



%
%
%
\emph{Software model-checkers}~\cite{Jhala2009a,Beyer2018} are tools for verifying systems described by programs.  They have been successfully applied to industrial software systems, in particular OS device drivers~\cite{Ball2002,Ball2004,Khoroshilov2010}.
To address the so-called ``state explosion problem'' in model checking real-world programs, the past decades have witnessed a significant development of state-space abstraction and refinement techniques including \emph{predicate abstraction}~\cite{Graf1997,Ball2001}, \emph{CounterExample-Guided Abstraction Refinement (CEGAR)}~\cite{Clarke2003a}, \emph{lazy abstraction}~\cite{Henzinger2002,Henzinger2004,McMillan2006,Beyer2012}, and \emph{Property Directed Reachability (PDR)}~\cite{Bradley2011,Een2011,Hoder2012,Cimatti2012,Cimatti2014,Komuravelli2013,Komuravelli2014,Birgmeier2014,Vizel2014,Beyer2020}.

Software model-checkers can be classified by their target programming language and target properties.
This paper mainly focuses on the safety verification problem for while languages.

The safety of a given program can be witnessed by an over-approximation of reachable states that does not intersect with the set of bad states.
Such a desirable over-approximation can be characterized by the following three conditions: (1) it contains all initial states, (2) it is closed under the transition relation, and (3) it does not contain any bad state.
Once a candidate of an approximation is given, it is relatively easy to check whether it satisfies the above conditions.
Hence the most challenging part of the software model-checking is to find an appropriate over-approximation.

A modern software model checker guesses a candidate of over-approximation and iteratively refines it until a desirable one is found.
Many papers were devoted to provide efficient procedures, including the above mentioned techniques for abstraction and refinement, to find an appropriate over-approximation, which vary in the structure of candidates and the candidate update method~\cite{Ball2001,Henzinger2002,Henzinger2004,McMillan2006,Hoder2012,Cimatti2012,Cimatti2014,Komuravelli2013,Komuravelli2014,Birgmeier2014}.

This paper aims to provide a unified account for a variety of approaches to software model-checking in terms of logic, or more precisely, proof search.
Since the notion of reachable states is an inductive notion, we need a proof system with (co-)induction.

A famous proof rules for reasoning about (co-)induction are based on pre- and/or post-fixed points (see, e.g., \cite{MartinLoef1971}). However, they are not appropriate for the purpose of interpreting various abstraction and refinement techniques as proof-search strategies, because these proof rules are applicable only after appropriate pre- or post-fixed-points are found.
Therefore the main process of software model-checking, trial-and-error search of appropriate over-approximation, has no logical interpretation.

We show that well-known procedures can naturally be seen as proof-search strategies of a non-standard proof system, known as a \emph{cyclic proof system}~\cite{Sprenger2003,Brotherston2011a}.
Figure~\ref{fig:illustration-of-the-coincidence} illustrates a correspondence between McMillan's \emph{lazy abstraction with interpolants}~\cite{McMillan2006} and cyclic proof search.
We do not explain the details here, which are the topic of Section~\ref{sec:IMPACT}, but we believe that the correspondence should be intuitively understandable.

Although the connection is somewhat expected, establishing a precise connection is not trivial.
This is because the most natural logical expression of the software model-checking, based on least fixed-point, does not fit to practical algorithms (see Section~\ref{sec:backward-symbolic-execution}).
Our key observation to establish a tight connection between cyclic proof search and software model-checking is that we use the dual notion, \emph{safe states}.

Our framework covers other important procedures including variants of \emph{PDR}~\cite{Hoder2012,Cimatti2012}, where we identify the significance of an underlying principle we call \emph{maximal conservativity} of refinement.  It turns out that the principle has also been adopted implicitly by highly efficient refinement algorithms, partially witnessing the usefulness of the principle in practice: one based on Craig interpolation~\cite{Craig1957a} for a variant of PDR~\cite{Vizel2014} and one based on quantified satisfaction~\cite{Farzan2016} for a solver of games over infinite graphs~\cite{Farzan2018}.  To the best of our knowledge, we are the first to apply the maximal conservativity, as one of the fundamental principles of software model checking, to formalize and compare variants of PDR and McMillan's lazy abstraction within a unified framework.  This allowed us to construct a counterexample for the refutational completeness of existing PDR variants~\cite{Hoder2012,Komuravelli2014,Komuravelli2016} and to obtain the first PDR variant with refutational completeness.

\begin{figure}[t]
    \centering
    \includegraphics[scale=0.20]{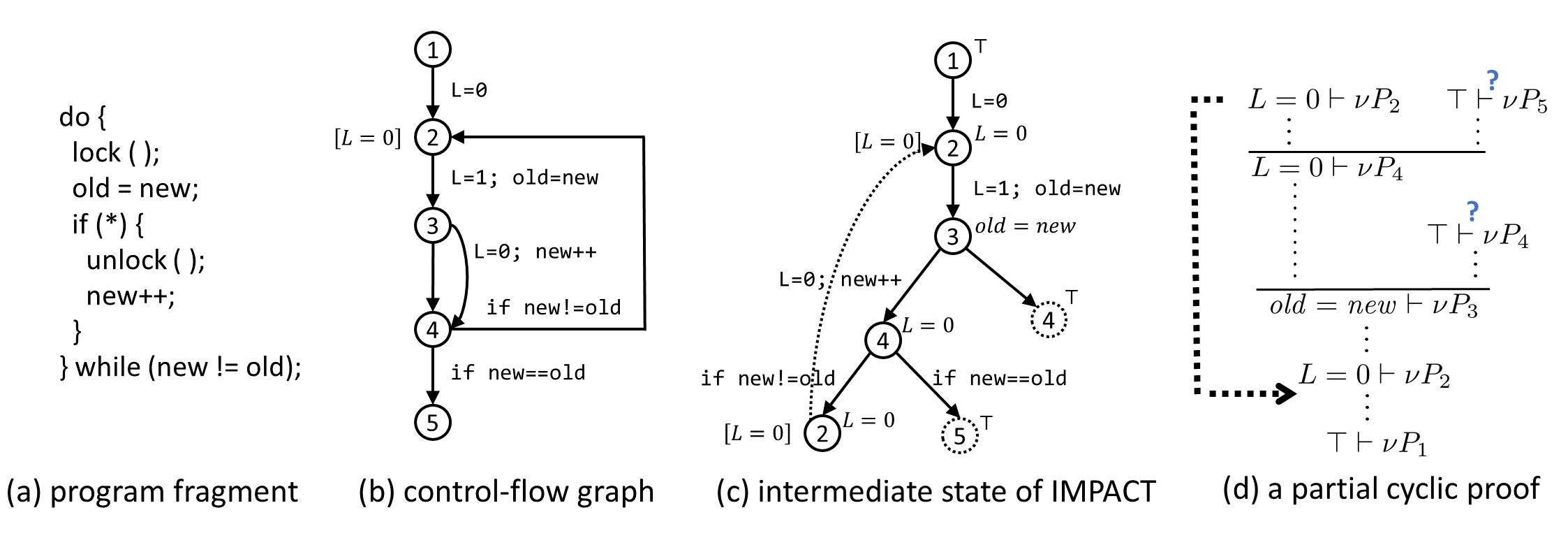}
    \caption{%
        Correspondence between an intermediate state of McMillan's \impact{} and an intermediate state of cyclic-proof search.
        Figures (a-c) originate from \cite{McMillan2006}.
        (a) An example of a program.
        (b) The control flow graph of (a).  The assertion is that \( L = 0 \) at location \( 2 \).
        (c) An intermediate state of \impact{}.  Nodes with the dotted line are needed to be expanded to confirm the safety.
        (d) An intermediate state of cyclic-proof search.  The sequents with ? mark are open, which are not yet proved.  For each location \( i \), \( \nu P_i \) is the formula corresponding to it; a sequent \( \varphi \vdash \nu P_i \) corresponds to the node in (c) of location \( i \) \huchanged{labeled} with the formula \( \varphi \).
    }
    \label{fig:illustration-of-the-coincidence}
\end{figure}
    
The logical foundation of software model checking based on cyclic proof search has the following advantages.
\begin{itemize}
    \item It allows us to compare different approaches to software model checking.  Different approaches are different proof-search strategies for the same goal sequent in the same proof system.
    It is even possible to compare a software model-checking algorithm with algorithms for other problems.
    As an example, we show that PDR shares the same heuristics based on the principle of maximal conservativity as the solver for infinite games~\cite{Farzan2018}. 
    \item It allows us to reconstruct well-known algorithms from a few simple principles.  This is because a proof system is demanding: use of a proof rule requires us to ensure side conditions.  For example, an internal state of PDR is a sequence of formulas that satisfies certain conditions; these conditions indeed coincide with the side conditions of proof rules.
    \item We obtain soundness of model-checking algorithms for free.  It is just a consequence of soundness of the proof system.
\end{itemize}
Note that the logical foundation goes beyond related ones~\cite{Flanagan2004,Podelski2007a,Bjorner2015a} based on constraint logic programming (CLP): they use CLP to capture mainly the process of \emph{modeling} verification problems, while we use a cyclic-proof system to capture both processes of \emph{modeling} and \emph{solving} verification problems; we expect that the logical foundation based on cyclic proofs would also be convenient for developing new algorithms and heuristics.
We shall briefly discuss some possible directions, though the detailed studies are left for future work.
Furthermore, we expect that the logical foundation paves the way for transferring abstraction and refinement techniques studied in the software model checking community to cyclic proof-search, in particular, for finding necessary cut-formulas, which is considered essential for the success of proof-search because cyclic proof systems usually do not admit cut-elimination~\cite{Kimura2020,Masuoka2021}.

\paragraph{Organization of this paper.}
The rest of the paper is organized as follows.  Section~\ref{sec:background} briefly reviews software model checking and two approaches: one is based on an over-approximation of reachable states and the other is based on a cyclic proof system.
Section~\ref{sec:idea} discusses what is an appropriate goal sequent to establish the connection between cyclic proof search and software model-checking.
Sections~\ref{sec:se}--\ref{sec:pdr} exemplify the generality of the logical foundation based on cyclic proofs, by presenting proof-search strategies corresponding to well-known software model-checking techniques: symbolic execution, bounded model checking, 
predicate abstraction, lazy abstraction, and PDR.  We discuss related work in Section~\ref{sec:related} and conclude the paper with final remarks in Section~\ref{sec:conc}.

%% file: background.tex
\section{Background}
\label{sec:background}

\subsection{Definition of Software Model-Checking}
Recall that we have focused on verification of safety properties of while programs.
A program induces a transition system with finite control states and infinite data, and we are interested in whether a bad state is reachable.
For simplicity, we assume that the target program has exactly one control state.
\tkchanged{It is not difficult to apply the ideas of the paper to programs with multiple but finite control \huchanged{states}, by introducing an auxiliary variable representing the program counter.}

A transition system consists of a domain \( \mathcal{D} \) of infinite data and a subset%
  \footnote{We shall identify a subset \( A \subseteq X \) with a predicate over \( X \) given by \( A(x) \Leftrightarrow (x \in A) \).}
\( \iota \subseteq \mathcal{D} \) of \emph{initial states} and a \emph{transition relation} \( \tau \subseteq \mathcal{D} \times \mathcal{D} \); an \emph{assertion} is given as a subset \( \alpha \subseteq \mathcal{D} \).
A typical example of \( \mathcal{D} \) is \( \mathbb{Z}^k \) where \( k \) is the number of variables in the program.

A state \( x \in \mathcal{D} \) is \emph{reachable} if there is a finite sequence \( x_0, x_1, \dots, x_n \) (\( n \ge 0 \)) such that \( \iota(x_0) \wedge \bigwedge_{i=1}^n \tau(x_{i-1}, x_i) \wedge x_n = x \) holds.
The \emph{safety verification problem} asks whether, given a while program and an assertion, there is a reachable state that violates the assertion.

\subsection{Approach Based on Over-approximation of Reachable States}\label{sec:pre:chc}
Unfortunately the set \( R \) of reachable states is hard to compute, and modern model-checkers avoid the exact computation of \( R \) by using an equational characterization of \( R \).
The set \( R \)
of reachable states satisfies the following condition:
\begin{equation}
    \iota(y) \vee (\exists x. R(x) \wedge \tau(x,y))
    ~~\Leftrightarrow~~
    R(y).
    \label{eq:reachable-states}
\end{equation}
Actually this condition characterizes the set of reachable states: it is the \emph{least} solution of \autoref{eq:reachable-states}.
Defining a notion as the least solution of an equation in a certain class is called an \emph{inductive definition}.
So \( R \) is an inductively defined predicate.

Thanks to the inductive definition of reachable states by \autoref{eq:reachable-states}, it suffices to find any solution \( R \) of \autoref{eq:reachable-states}, which is not necessarily least, such that \( R(x) \Rightarrow \alpha(x) \).
If any solution \( R \) of \autoref{eq:reachable-states} satisfies \( R(x) \Rightarrow \alpha(x) \), then obviously so does the least solution of \autoref{eq:reachable-states}.

The condition can be further relaxed.
It is well-known that the least solution of \autoref{eq:reachable-states} coincides with the least solution of
\begin{equation}
    \iota(y) \vee (\exists x. R(x) \wedge \tau(x,y))
    ~~\Rightarrow~~
    R(y).
    \label{eq:reachable-states-ineq}
\end{equation}
Theoretically the least solution can be computed by iteratively applying the predicate transformer \( \mathcal{F}[\varphi](y) := \iota(y) \vee (\exists x. \varphi(x) \wedge \tau(x,y)) \) to the empty set.

\subsection{Approach Based on a Proof System for Inductive Predicates}
The previous subsection reduces the safety verification problem to the entailment problem \( R(x) \Rightarrow \alpha(x) \), where \( R \) is an inductively defined predicate.
Hence a proof system for inductively defined predicates would serve as a basis for software model-checking.

Let us first fix notations.
Assume an equation \( P(x) \Leftrightarrow \delta[P](x) \), where \( P \) occurs only positively in \( \delta[P] \) \tkchanged{(i.e.~every occurrence of \( P \) \huchanged{is} under an even number of negation operators)}.
Here \( \delta \) is an unary predicate that defines \( P \), which itself may depend on \( P \); this dependency is written explicitly as \( \delta[P] \).
We write the least solution of this equation as \( (\mu P)(x) \),%
  \footnote{To simplify the notation, the defining equation is implicit.  By explicitly writing the defining equation, \( \mu P \) can be written as \( \mu P. \lambda x. \delta[P](x) \).}
where \( \mu \) is the standard symbol for the least fixed-points.

This paper focuses on a \emph{cyclic proof system}~\cite{Sprenger2003,Brotherston2011a}.
It is an extension of the sequent calculus for the first-order classical logic, with two additional mechanisms to deal with inductively defined predicates.
The first one is the additional proof rules for inductively defined predicates: 
\begin{equation*}
    \infer[\textsc{($\mu$-L)}]{
        (\mu P)(x) \vdash \varphi(x)
    }{
        \delta[(\mu P)](x) \vdash \varphi(x)
    }
    \qquad\qquad
    \infer[\textsc{($\mu$-R)}]{
        \varphi(x) \vdash (\mu P)(x)
    }{
        \varphi(x) \vdash \delta[(\mu P)](x)
    }
\end{equation*}
\huchanged{These rules just expand} the definition of the predicate \( \mu P \).
For example, for the reachability predicate \( \mu R \) defined by \autoref{eq:reachable-states}, we have
\begin{equation*}
    \infer[\textsc{($\mu$-L)}.]{
        (\mu R)(x) \vdash \varphi(x)
    }{
        \iota(x) \vee (\exists y. (\mu R)(y) \wedge \tau(y,x)) \vdash \varphi(x)
    }
\end{equation*}
A characteristic feature of the cyclic proof system is that a proof is not a tree but can have cycles; instead of proving the sequent of a leaf node, one can make a link to its ancestor with the same sequent.

A \emph{pre-proof} is a proof-like tree whose leaves are either instances of the axiom rule or equipped with links to their ancestors.
\autoref{fig:global-trace-condition} shows two examples of pre-proofs.
Intuitively a link corresponds to a use of induction; for example, the proof of \autoref{fig:global-trace-condition}(a) says that ``we prove the sequent \( (\mu R)(x) \vdash \varphi(x) \) at the second line from the bottom by induction on \( x \)" and that ``the leaf node \( (\mu R)(y) \vdash \varphi(x) \) follows from the induction hypothesis.''

Not all pre-proofs are valid; for example, the proof in \autoref{fig:global-trace-condition}(b) has an invalid conclusion.
Intuitively the invalidity is caused by a wrong use of the induction hypothesis.
Recall that the induction hypothesis is applicable to only smaller elements, i.e.~when proving \( P(t) \), the induction hypothesis tells us \( P(u) \) for \( u < t \) but not for \( u \ge t \).
So if one applies the induction hypothesis to an element not smaller than \( t \), it is a wrong proof.
\autoref{fig:global-trace-condition}(b) makes this kind of mistake.

A pre-proof is valid if it uses the ``induction hypotheses'' appropriately.
The appropriateness in the context of cyclic proof systems is called the \emph{global trace condition}.
Because its definition is rather complicated (see~\cite{Brotherston2011a}), we do not formally define the global trace condition.
For pre-proofs appearing in this paper, the following condition is equivalent to the global trace condition:
\begin{quote}
    If there is a link from a leaf \( l \) to its ancestor \( n \), \textsc{($\mu$-L)} rule is used in the path from \( n \) to \( l \).
    \label{sec:gtc}
\end{quote}
\autoref{fig:global-trace-condition}(a) satisfies the global trace condition in the above sense, but \autoref{fig:global-trace-condition}(b) does not.
A \emph{proof} is a pre-proof that satisfies the global trace condition.

\begin{figure}[t]
    \centering
    \includegraphics[scale=0.4]{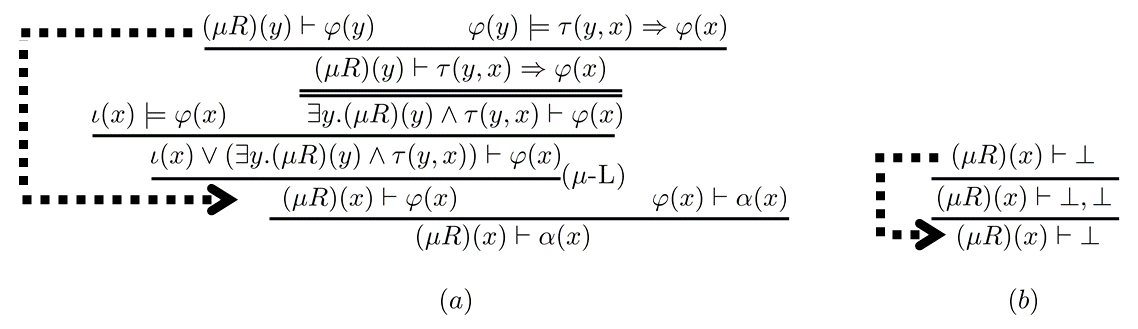}
    \vspace{-10pt}
    \caption{Pre-proofs and the global trace condition.
    The preproof (a) satisfies the global trace condition, but the preproof (b) does not.
    So (a) is a cyclic proof but (b) is not.
    The inductive predicate \( R \) is defined by \autoref{eq:reachable-states}.}
    \label{fig:global-trace-condition}
\end{figure}

\begin{example}
    Suppose that \( \varphi \) is a solution of \autoref{eq:reachable-states} such that \( \varphi(x) \Rightarrow \alpha(x) \).
    Then the cyclic proof of \autoref{fig:global-trace-condition}(a) shows the sequent corresponding to the safety of the program.
    \qed
\end{example}

\subsection{Goal-oriented proof search}
\tkchanged{This paper studies a connection between software model-checking and the proof-search problem for a cyclic proof system, which asks to construct a proof of a given sequent.}
Proof search strategies of this paper are goal-oriented, i.e.~a proof is constructed in a bottom-up manner.
So an intermediate state is \removed{thus} a proof-like structure in which some leaves are not yet proved.
We call such an unproved leaf sequent an \emph{open sequent} and a proof-like structure with open leaves a \emph{partial proof}.

We fix a signature and a theory and assume that \( \iota \), \( \tau \) and \( \alpha \) are formulas of the signature.
We assume an external solver of the theory.
We use \( \models \) instead of \( \vdash \) for the provability of a sequent handled by the external solver.
For example, if a proof rule contains \( \exists x. \iota(x) \wedge \tau(x,y) \models \alpha(y) \) as a premise, its validity should be checked by the external solver; our proof search strategies do not give any details of how to prove or disprove such sequents.
We assume that a partial proof has no open sequent of the form \( \varphi \models \psi \).

A partial proof is \emph{valid} if all \( \models \)-sequents in it are valid.
We sometimes think of an \emph{invalid partial proof}, which contains a false \( \models \)-sequent.
A partial proof is either valid or invalid.

%% file: basic.tex
\section{Basic Notions of Software Model-Checking}
\label{sec:base}

\subsection{The Goal Sequent}
\label{sec:idea}

As the first step to establish a connection between software model-checking and cyclic-proof search, this section describes the goal sequent corresponding to software model-checking.
%
It might seem natural to use the sequent
\begin{equation*}
    (\mu R)(y) \vdash \alpha(y)
    \qquad\mbox{where}~
    R(y) \Leftrightarrow \big(\iota(y) \vee (\exists x. R(x) \wedge \tau(x,y))\big),
\end{equation*}
which directly expresses that all reachable states satisfy the assertion \( \alpha \).
Unfortunately proof search for this sequent does not precisely correspond to existing procedures.

The key observation of this paper is that we should employ another logically-equivalent sequent, which uses the dual notion.
We explain the dual notion from the view point of program verification, introducing the notion of \emph{safe states}.
We also briefly explain a proof system that handles \emph{co-inductively} defined predicates, the duals of inductively defined ones.





\paragraph{Safe states}
From the view point of program verification, the dual formalization can be obtained by using the notion of \emph{safe states}.
A state is \emph{safe} if the execution from the state will never reach a bad state.
The safety verification problem can be alternatively formalized as the problem asking whether all initial states are safe.

Let us formally define safe states.
A state \( x \) is \emph{unsafe} if there exists a sequence \( x_0, \dots, x_n \) (\( n \ge 0 \)) such that \( x = x_0 \wedge \bigwedge_{i=1}^{n} \tau(x_{i-1}, x_i) \wedge \neg \alpha(x_n) \).
A state is \emph{safe} if it is not unsafe.
The set \(S\) of safe states has an equational characterization similar to the set of reachable states:
it is the \emph{greatest} solution of
\begin{equation}
    \alpha(x) \wedge (\forall y. \tau(x,y) \Rightarrow S(y))
    ~~\Leftrightarrow~~
    S(x).
    \label{eq:defining-equation-for-safe-states}
\end{equation}
This kind of definition, specifying something as the greatest solution of an equation, is called \emph{coinductive definition}, which is the dual of inductive definition.
We write \( \nu P \) for the greatest solution (provided that the equation for \( P \) is understood).
Following this notation, the set of safe states is \( \nu S \).

Now the safety verification is reduced to the validity of the following sequent:
\begin{equation*}
    \iota(x) \vdash (\nu S)(x)
    \qquad\mbox{where}~
    S(x) \Leftrightarrow \big(\alpha(x) \wedge (\forall y. \tau(x,y) \Rightarrow S(y))\big).
\end{equation*}
This paper shows a tight connection between software model-checkers and cyclic-proof search for the above sequent.

\paragraph{Cyclic proofs for coinductively defined predicates}
The idea of cyclic proof systems is applicable to coinductively defined predicates.
The proof system has the following rules:
\begin{equation*}
    \infer[\textsc{($\nu$-L)}]{
        (\nu P)(x) \vdash \varphi(x)
    }{
        \delta[(\nu P)](x) \vdash \varphi(x)
    }
    \qquad\qquad
    \infer[\textsc{($\nu$-R)}]{
        \varphi(x) \vdash (\nu P)(x)
    }{
        \varphi(x) \vdash \delta[(\nu P)](x)
    }
\end{equation*}
where \( P(x) \Leftrightarrow \delta[P](x) \) is the defining equation for \( P \).
A proof is a preproof that satisfies the following condition:
\begin{quote}
    If there is a link from a leaf \( l \) to its ancestor \( n \), \textsc{($\nu$-R)} rule is used in the path from \( n \) to \( l \).    
\end{quote}

\subsection{Symbolic Execution}
\label{sec:se}
\tkchanged{The \textsc{($\nu$-R)} rule with some postprocessing corresponds to the symbolic execution of the one-step transition}.
This observation, which can be easily shown, is quite important because it is why software model-checking processes coincide with proof-search strategies of \( \iota(x) \vdash (\nu P)(x) \).
It also tells us why \( (\mu P)(x) \vdash \alpha(x) \), which seems a more natural expression of the model-checking problem, does not suit for the purpose.

\begin{definition}[Symbolic execution]
    Let \( \varphi(x) \) be a formula denoting a set of states.
    Then the set of states after the transition \( \tau \) from \( \varphi \) can be represented by
    \(
        \varphi'(y) := \exists x. \varphi(x) \wedge \tau(x,y)
    \),
    and we write \( \varphi \rightsquigarrow \varphi' \).
    This transition on formulas is called \emph{symbolic execution}.
    \qed
\end{definition}

Let us examine the connection between symbolic execution and \textsc{($\nu$-R)} rule.
We discuss a proof-search strategy for sequents of the form \( \varphi(x) \vdash (\nu P)(x) \), which slightly generalize the goal sequent \( \iota(x) \vdash (\nu P)(x) \).

Applying the \textsc{($\nu$-R)} rule to the sequent \( \varphi(x) \vdash (\nu P)(x) \), we obtain
\begin{equation*}
    \infer[\textsc{($\nu$-R)}.]{
        \varphi(x) \vdash (\nu P)(x)
    }{
        \varphi(x) \vdash \big(\forall y. \tau(x,y) \Rightarrow (\nu P)(y)\big) \wedge \alpha(x)
    }
\end{equation*}
Notice that the new open sequent, which is the premise of the above partial proof, has a shape much different from the goal sequent.
For the purpose of proof search, it is convenient to restrict the shape of open sequents.
We introduce the first proof-search principle of this paper.
\begin{principle}\label{principle:shape-of-goals}
    Try to fit the shape of open sequents into the form \( \varphi(x) \vdash (\nu P)(x) \).
\end{principle}

\tkchanged{We see that this principle leads to a proof strategy that corresponds to symbolic execution.}
Following Principle~\ref{principle:shape-of-goals}, we try to simplify \tkchanged{the open sequent \( \varphi(x) \vdash \big(\forall y. \tau(x,y) \Rightarrow (\nu P)(y)\big) \wedge \alpha(x) \)} of the above partial proof.
Fortunately its validity is equivalent to the validity of two simple sequents as the following partial proof shows:
%
\begin{equation*}
    \infer[\textsc{($\nu$-R)}]{
        \varphi(x) \vdash (\nu P)(x)
    }{
        \infer{
            \varphi(x) \vdash \big(\forall y. \tau(x,y) \Rightarrow (\nu P)(y)\big) \wedge \alpha(x)
        }{ 
            \infer{
                \varphi(x) \vdash \forall y. \tau(x,y) \Rightarrow (\nu P)(y)
            }{
                \infer{
                    \varphi(x) \vdash \tau(x,y) \Rightarrow (\nu P)(y)
                }{
                    \infer[\textsc{(Cut)}]{
                        \varphi(x), \tau(x,y) \vdash (\nu P)(y)
                    }{
                        \infer{
                            \varphi(x), \tau(x,y) \vdash \exists x. \varphi(x) \wedge \tau(x,y)
                        }{
                            \vdots
                        }
                        &&
                        \exists x. \varphi(x) \wedge \tau(x,y) \vdash (\nu P)(y)
                    }
                }
            }
            &&
            \varphi(x) \models \alpha(x)
        }
    }
\end{equation*}
(where the proof of \( \varphi(x), \tau(x,y) \vdash \exists x. \varphi(x) \wedge \tau(x,y) \) is omitted).
One of the \huchanged{leaves} is an open sequent of the required form, and the other is fixed-point free and passed to an external solver.
So this partial proof meets the criterion of \huchanged{Principle~\autoref{principle:shape-of-goals}}.

The above argument shows that the rule
\begin{equation*}
    \infer[\textsc{(SE)}]{
        \varphi(x) \vdash (\nu P)(x)
    }{
        \exists x. \varphi(x) \wedge \tau(x,y) \vdash (\nu P)(y)
        &&
        \varphi(x) \models \alpha(x)
    }
\end{equation*}
is a derived rule.
We call it the (precise) symbolic execution rule \textsc{(SE)}.
Since we are interested in proof-search strategies, this rule should be read bottom-up: In order to prove the satefy of \( \varphi \), it suffices to (1) check if \( \varphi \) does not \huchanged{violate} the assertion, (2) do the symbolic execution \( \varphi \rightsquigarrow \varphi' \), and (3) show the safety of \( \varphi' \) by proving the sequent \( \varphi'(x) \vdash (\nu P)(x) \).

\subsection{Bounded Model-Checking}
\label{sec:bmc}
The bounded model-checking problem~\cite{Biere1999} asks, given a bound \( k \) and an initial state \( \varphi_0 \), whether a bad state is reachable from a state satisfying \( \varphi_0 \) by transitions of steps \( \le k \).
More formally, it checks if \( \varphi_0 \rightsquigarrow \varphi_1 \rightsquigarrow \dots \rightsquigarrow \varphi_k \) implies \( \varphi_i(x) \models \alpha(x) \) for every \( i = 0, \dots, k \).
If it is the case, we say \( \varphi_0 \) is \emph{safe within \( k \) steps}.
The bounded model-checking problem is important because it is decidable in certain settings, and software model-checkers often have subprocedures corresponding to bounded model-checking.

In terms of cyclic-proof search, bounded model-checking is to construct a valid partial proof consisting only of \textsc{(SE)}.
By iteratively applying \textsc{(SE)}, we obtain a (valid or invalid) partial proof in \autoref{fig:bounded-model-checking}.
Here \( \varphi_0 \rightsquigarrow \varphi_1 \rightsquigarrow \dots \rightsquigarrow \varphi_k \rightsquigarrow \varphi_{k+1} \).
This is a valid partial proof if \( \varphi_i(x) \models \alpha(x) \) for every \( i = 0, \dots, k \).
Therefore bounded model-checking with bound \( k \) is equivalent to the problem asking if the \( k \) consecutive applications of \textsc{(SE)} to the corresponding sequent is a valid partial proof.
\begin{figure}[t]
    \begin{equation*}
        \infer[\textsc{(SE)}]{
            \varphi_0(x) \vdash (\nu P)(x)
        }{
            \infer[\textsc{(SE)}]{
                \varphi_1(x) \vdash (\nu P)(x)
            }{
                \infer*{
                    \varphi_2(x) \vdash (\nu P)(x)
                }{
                    \infer[\textsc{(SE)}]{
                        \varphi_k(x) \vdash (\nu P)(x)
                    }{
                        \varphi_{k+1}(x) \vdash (\nu P)(x)
                        &&
                        \varphi_k(x) \models \alpha(x)
                    }
                }
                \hspace{-60pt}
                &&
                \varphi_1(x) \models \alpha(x)
            }
            &&
            \varphi_0(x) \models \alpha(x)
        }
    \end{equation*}
    \vspace{-20pt}
    \caption{%
        A partial proof corresponding to bounded model-checking.
        Here \( \varphi_0 \rightsquigarrow \varphi_1 \rightsquigarrow \dots \rightsquigarrow \varphi_k \rightsquigarrow \varphi_{k+1} \).
        The partial proof is valid if \( \varphi_i(x) \models \alpha(x) \) for every \( i = 0, \dots, k \).}
    \label{fig:bounded-model-checking}
\end{figure}
\begin{proposition}
    The \( k \) consecutive applications of \textsc{(SE)} to \( \varphi(x) \vdash (\nu X)(x) \) is a valid partial proof if and only if \( \varphi \) is safe within \( k \) steps.
    \qed
\end{proposition}

\paragraph{Forward Criterion and States of Distance $k$}
Sometimes postprocessing the result of bounded model-checking gives a proof of unbounded safety.
We call the criterion in \cite{Sheeran2000} based on a notion of distance (called \emph{forward diameter} in \cite{Sheeran2000}) \emph{forward criterion}.

Let us fix the set \( \iota \) of initial states.
A reachable state \( x \) is of \emph{distance} \( k \) if there exists a sequence \( x_0, \dots, x_k \) witnessing the reachability of \( x \) from \( \iota \) (i.e.~\( \iota(x_0) \wedge \bigwedge_{i=1}^{k} \tau(x_{i-1}, x_i) \wedge (x_k = x) \)) and there is no shorter witness.
If \( \iota = \varphi_0 \rightsquigarrow \varphi_1 \rightsquigarrow \dots \rightsquigarrow \varphi_k \), then the set of states of distance \( k \) is represented by \( \varphi_k(x) \wedge \bigwedge_{i=0}^{k-1} \neg \varphi_i(x) \).
The \emph{forward criterion} says that a program is safe if the program is safe with bound \( k \) and the set of states of distance \( k+1 \) is empty.

The set of distance-\( k \) states naturally appears in cyclic proof search.
Recall that, in a cyclic proof system, a leaf sequent can be proved by putting a link to its ancestor \huchanged{labeled} by the same judgement.
For example, a proof of the open sequent \( \varphi_{k+1}(x) \vdash (\nu P)(x) \) in \autoref{fig:bounded-model-checking} is able to use the assumptions \( \varphi_i(x) \vdash (\nu P)(x) \) (\( 0 \le i < k+1 \)).
The set of distance-\( k \) states naturally arises when we try to exploit the assumptions to the full.

The idea is to split \( \varphi_{k+1}(x) \) into \( \varphi_i(x) \) (\( 0 \le i < k+1 \)) and the remainder, say \( \psi(x) \).
Assume a \tkchanged{formula} \( \psi \) such that \( \varphi_{k+1}(x) \models \psi(x) \vee \bigvee_{0 \le i < k+1} \varphi_i(x) \).
Let us write \( \Psi(x) \) for \( \psi(x) \vee \bigvee_{0 \le i < k+1} \varphi_i(x) \).
The open sequent in the partial proof in \autoref{fig:bounded-model-checking} can be expanded as
\begin{equation*}
    \infer[\textsc{(Cut)},]{
        \varphi_{k+1}(x) \vdash (\nu P)(x)
    }{
        \varphi_{k+1}(x) \models \Psi(x)
        &&
        \hspace{-10pt}
        \infer[\textsc{($\vee$-L)}\times (k+1)]{
            \Psi(x) \vdash (\nu P)(x)
        }{
            \psi(x) \vdash (\nu P)(x)
            &&
            \varphi_0(x) \vdash (\nu P)(x)
            &&
            \hspace{-10pt}
            \dots
            \hspace{-10pt}
            &&
            \varphi_k(x) \vdash (\nu P)(x)
        }
    }
\end{equation*}
which has \( k+1 \) open sequents.
All but one open sequent can be immediately proved by links to ancestors; the only remaining open sequent is \( \psi(x) \vdash (\nu P)(x) \).
Since \( \psi \) appears in a contra-variant position, a stronger \( \psi \) is more desirable.
The strongest \( \psi \) that satisfies the requirement is given by
\(
    \psi(x)
    \Leftrightarrow
    \Big( \varphi_{k+1}(x) \wedge \bigwedge_{0 \le i < k+1} \neg \varphi_i(x) \Big)
\).
This is exactly the set of distance-\( (k+1) \) states.
In particular, if the forward criterion holds, one can choose \( \psi(x) \) as \( \bot \).
In this case, the remaining open sequent is \( \bot \vdash (\nu P)(x) \), which is trivially provable independent of \( P \); hence we complete the proof of \( \varphi_0(x) \vdash (\nu P)(x) \).\footnote{
    Actually what we obtained is a pre-proof, i.e.~a proof-like structure with cycles, and we should check the global \huchanged{trace} condition to conclude that it is a proof.
    To see that the pre-proof satisfies the global trace condition, recall that \textsc{(SE)} is a derived rule using \textsc{($\nu$-R)} rule.
}

\subsection{Predicate Abstraction}
\label{sec:pa}

As shown in the previous section, \textsc{(Cut)} plays an important role in the construction of a cyclic proof.
Let us introduce the derived rule
\begin{equation*}
    \infer[\textsc{(SE+Cut)},]{
        \varphi(x) \vdash (\nu P)(x)
    }{
        \psi(x) \vdash (\nu P)(x)
        &&
        \exists y. \varphi(y) \wedge \tau(y,x) \models \psi(x)
        &&
        \varphi(x) \models \alpha(x)
    }
\end{equation*}
which can be proved by
\begin{equation*}
    \infer[\textsc{(SE)}.]{
        \varphi(x) \vdash (\nu P)(x)
    }{
        \infer[\textsc{(Cut)}]{
            \exists x. \varphi(x) \wedge \tau(x,y) \vdash (\nu P)(y)
        }{
            \exists x. \varphi(x) \wedge \tau(x,y) \models \psi(y)
            &&
            \psi(x) \vdash (\nu P)(x)
        }
        &&
        \varphi(x) \models \alpha(x)
    }
\end{equation*}
We call \( \psi \) in \textsc{(SE+Cut)} the \emph{cut formula}.

The choice of the cut formula is quite important.
If it is too strong, it is not an invariant; if it is too weak, the sequent \( \psi(x) \vdash (\nu P)(x) \) becomes invalid.

\emph{Predicate abstraction}~\cite{Graf1997,Ball2001} is a heuristic for the choice of the cut formula.
The idea is to further restrict the shape of open sequents: an open sequent must be of the form \( \vartheta(x) \vdash (\nu P)(x) \) where \( \vartheta \) comes from a finite set of predicates, say \( \Xi \).
Suppose that \( \Xi \) is closed under conjunction.
In order to avoid making the open sequent invalid, the best choice of the cut formula is the strongest formula in \( \Xi \).

This idea leads to the following principle.
\begin{principle}\label{principle:strongest-cut}
    Use \textsc{(SE+Cut)}.
    Choose the strongest formula \( \psi \in \Xi \) satisfying \( \exists x. \varphi(x) \wedge \tau(x,y) \models \psi(\huchanged{y}) \) as the cut formula.
\end{principle}

The proof-search procedure following Principles~\ref{principle:shape-of-goals} and \ref{principle:strongest-cut} terminates.
Since \( \Xi \) is finite,
the length of the path in a partial proof is bounded by the number of candidate formulas; if it exceeds, then a sequent appears twice in the path and thus putting a link completes the proof.

%% file: impact.tex
\section{Lazy Abstraction with Interpolants: \impact{}}
\label{sec:IMPACT}

This section explains another proof-search strategy that corresponds to the software model-checker \impact{} developed by McMillan~\cite{McMillan2006}.
For readers familiar with \impact{}, Table~\ref{table:impact-cyclic-proof-dictionary} gives a dictionary of terminologies of \impact{}~\cite{McMillan2006} and cyclic-proof search and \autoref{fig:impact-partial-proof} illustrates the correspondence between internal states of \impact{} (called \emph{program unwinding}) and partial proofs.

\begin{table}[t]
    \centering
    \begin{tabular}{l|ll}
        \impact{} & & Cyclic proof \\ \hline
        program unwinding & {}\quad{} & partial proof
        \\
        \qquad safe, well-\huchanged{labeled} --- & & \qquad valid ---
        \\
        \qquad safe, well-\huchanged{labeled}, complete --- \qquad {} & & \qquad cyclic proof
        \\
        covering relation & & collection of links
        \\
        \textsc{Expand} procedure & & application of \textsc{(SE+Cut)}
    \end{tabular}
    \vspace{5pt}
    \caption{An \impact{}-Cyclic-proof dictionary.}
    \vspace{-15pt}
    \label{table:impact-cyclic-proof-dictionary}
\end{table}

\begin{figure}[t]
    \begin{equation*}
        \makebox[-70pt][0pt]{}
        \infer{
            \iota(x) \vdash (\nu P)(x)
        }{
            \infer{
                \varphi_1(x) \vdash (\nu P)(x)
            }{
                \infer*{
                    \varphi_2(x) \vdash (\nu P)(x)
                }{
                    \makebox[70pt][0pt]{}
                    \infer{
                        \varphi_n(x) \vdash (\nu P)(x)
                        \makebox[70pt][0pt]{}
                    }{
                        \varphi_{n+1}(x) \vdash (\nu P)(x)
                        &&
                        \raisebox{1em}{\oalign{
                            \hfill $\varphi_n(x) \models \alpha(x)\;\;\!\!$
                            \crcr
                            $\exists x. \varphi_n(x) \wedge \tau(x,y) \models \varphi_{n+1}(y)$
                        }}
                    }
                }
                &&
                \hspace{-100pt}
                \raisebox{1em}{\oalign{
                    \hfill $\varphi_1(x) \models \alpha(x)~\,$
                    \crcr
                    $\exists x. \varphi_1(x) \wedge \tau(x,y) \models \varphi_2(y)$
                }}
            }
            &&
            \hspace{-10pt}
            \raisebox{1em}{\oalign{
                \hfill $\iota(x) \models \alpha(x)~\,$
                \crcr
                $\exists x. \iota(x) \wedge \tau(x,y) \models \varphi_1(y)$
            }}
        }
    \end{equation*}
    \vspace{-10pt}
    \caption{%
        A partial proof (above) and the corresponding internal state of \impact{} (below).
        All rules used in the partial proof are \textsc{(SE+Cut)}.}
    \label{fig:impact-partial-proof}
\end{figure}

A significant difference from the previously discussed proof-search strategies is that \impact{} does not hesitate to generate an invalid sequent.
The difference allows us to simplify the strategy of the cut-formula choice.
Recall that the only requirement for the cut formula \( \psi \) in \textsc{(SE+Cut)} rule is \( \exists x. \varphi(x) \wedge \tau(x,y) \vdash \psi(y) \), and no upper bound is given.
\impact{} tentatively chooses \( \top \) as the cut formula.
\begin{principle}\label{principle:optimistic-cut}
    Use \textsc{(SE+Cut)}.
    Tentatively choose \( \top \) as the cut formula.
\end{principle}

Of course, \( \top \) is usually too weak and the open sequent \( \top \vdash (\nu P)(x) \) is often invalid.
The invalidity of the sequent is revealed by applying \textsc{(SE+Cut)} again:
\begin{equation*}
    \infer
    {
        \iota(x) \vdash (\nu P)(x)
    }{
        \infer
        {
            \top \vdash (\nu P)(x)
        }{
            \top \vdash (\nu P)(x)
            &&
            \hspace{-12pt}
            \exists x. \top \wedge \tau(x,y) \models \top
            &&
            \hspace{-12pt}
            \top \models \alpha(x)
        }
        &&
        \hspace{-15pt}
        \exists x. \iota(x) \wedge \tau(x,y) \models \top
        &&
        \hspace{-12pt}
        \iota(x) \models \alpha(x)
    }
\end{equation*}
This partial proof is invalid as \( \top \models \alpha(x) \) does not hold (otherwise the verification problem is trivial).
\impact{} tries to fix an invalid partial proof by strengthening the cut formulas, following the next principle.\footnote{
    \tkchanged{Here we note some differences between \impact{} and the proof-search strategy following Principle~\ref{principle:refinement}.}
    See Remark~\ref{rem:impact-diff}.
}
\begin{principle}\label{principle:refinement}
    When the partial proof is found to be invalid,
    replace each cut formula \( \varphi_i(x) \) with \( \varphi_i(x) \wedge Q_i(x) \), where \( Q_i \) is a fresh predicate variable, and solve the constraints on \( Q_i \).
\end{principle}

If the current partial proof is that in \autoref{fig:impact-partial-proof}, cut-formulas \( \varphi_1, \dots, \varphi_n \) must be strengthened if \( \varphi_n(x) \nvDash \alpha(x) \).
For simplicity, we assume \( \varphi_i(x) \models \alpha(x) \) for \( i < n \).
We replace \( \varphi_i(x) \) in the partial proof with \( \varphi_i(x) \wedge Q_i(x) \) (\( 1 \le i \le n \)), where \( Q_1, \dots, Q_n \) are distinct predicate variables.
The side conditions \( \varphi_i(x) \models \alpha(x) \) and \( \exists x. \varphi_i(x) \wedge \tau(x,y) \models \varphi_{i+1}(x) \) are transformed into \( \varphi_i(x) \wedge Q_i(x) \models \alpha(x) \) and \( \exists x. \varphi_i(x) \wedge Q_i(x) \wedge \tau(x,y) \models \varphi_{i+1}(y) \wedge Q_{i+1}(y) \).
Slightly simplifying the constraints by using \( \exists x. \varphi_i(x) \wedge \tau(x,y) \models \varphi_{i+1}(y) \) (\( i = 0, \dots, n \)) and \( \varphi_i(x) \models \alpha(x) \) (\( i = 0, \dots, n-1 \)), we obtain the following constraint set:
\begin{align*}
    \iota(x_0) \wedge \tau(x_0,x_1) &~~\Rightarrow~~ Q_1(x_1), \\
    Q_i(x_i) \wedge \big(\varphi_i(x_i) \wedge \tau(x_i,x_{i+1})\big) &~~\Rightarrow~~ Q_{i+1}(x_{i+1}),
    &\qquad
    i = 1, \dots, n-1 \\
    Q_{n}(x_n) \wedge \varphi_n(x_n) &~~\Rightarrow~~ \alpha(x_n).
\end{align*}
If a solution \( (\psi_i)_{i = 1,\dots,n} \) for \( (Q_i)_i \) is given, we can fix the partial proof by replacing \( \varphi_i \) with \( \varphi_i \wedge \psi_i \).
A solution of the constraints is known as an \emph{interpolant}.
\begin{definition}
    Given a sequence \( \vec{\vartheta} = \vartheta_1 \vartheta_2 \dots \vartheta_{n+1} \) of formulas,
    its \emph{interpolant} is a sequence \( \vec{\psi} = \psi_0 \psi_1 \dots \psi_{n+1} \) of formulas that satisfies the following conditions:
    (1) \( \psi_0 = \top \) and \( \psi_{n+1} = \bot \),
    (2) \( \models \psi_{i-1} \wedge \vartheta_i \Rightarrow \psi_i \) for every \( 1 \le i \le n \), and
    (3) \( \mathrm{fv}(\psi_i) \subseteq \mathrm{fv}(\vartheta_1,\dots,\vartheta_i) \cap \mathrm{fv}(\vartheta_{i+1}, \dots, \vartheta_{n+1}) \) for every \( 1 \le i \le n \).
    \qed
\end{definition}
A solution of the above constraint is an interpolant of \( \vec{\vartheta} \) where \( \vartheta_1 := \iota(x_0) \wedge \tau(x_0, x_1) \), \( \vartheta_{i+1} := \varphi_i(x_i) \wedge \tau(x_i, x_{i+1}) \) for \( i = 1,\dots,n-1 \), and \( \vartheta_{n+1} = \varphi_n(x_n) \wedge \neg \alpha(x_n) \).

The termination condition of \impact{} is also closely related to cyclic proofs.
\impact{} terminates if \( \varphi_i(x) \models \varphi_j(x) \) for some \( i > j \).
Then replacing the partial proof above \( \varphi_i(x) \vdash (\nu P)(x) \) with
\begin{equation*}
    \infer[\textsc{(Cut)}]{
        \varphi_i(x) \vdash (\nu P)(x)
    }{
        \varphi_i(x) \models \varphi_j(x)
        &&
        \varphi_j(x) \vdash (\nu P)(x)
    }
\end{equation*}
and making a link from \( \varphi_j(x) \vdash (\nu P)(x) \) to the ancestor yields a cyclic proof.

In summary, \impact{} can be seen as a cyclic-proof search strategy following Principles~\ref{principle:shape-of-goals}, \ref{principle:optimistic-cut} and \ref{principle:refinement}.

\begin{remark}\label{rem:impact-diff}
    We discuss differences between \impact{} and the proof-search strategy following Principles~\ref{principle:shape-of-goals}, \ref{principle:optimistic-cut} and \ref{principle:refinement}.
    \tkchanged{First the proof-search strategy does not cover an important heuristic, known as \emph{forced covering}, which tries to refine predicates so that one can make a cycle from a certain node to another;
    this feature can be understood as an additional proof-search rule}.
    \tkchanged{Second \impact{} computes an interpolant of a sequence slightly different from \( \vec{\vartheta} \)}:
    the sequence \( \vec{\vartheta} \) contains the information of the current cut-formulas \( \varphi_i \).
    Theoretically, the difference does not affect the refinement process (modulo logical equivalence), since the interpolants are conjoined with the current cut-formulas.  In practice, however, \cite{Vizel2014} claims that the use of the current cut-formulas is one of the keys to achieve their \emph{incremental} interpolation-based model checking, which indeed shows an improvement over previous methods based on interpolation a la \impact{}.
    \tkchanged{Third the actual implementation of\/ \impact{} uses a special kind of interpolant, which follows Principle~\ref{principle:maximally-conservative} below as well}.\footnote{\tkchanged{Personal communication}.}
    \qed
\end{remark}

%


\begin{remark}
    One may notice that the combination of Principles~\ref{principle:shape-of-goals}, \ref{principle:strongest-cut} and \ref{principle:refinement} (instead of \ref{principle:shape-of-goals}, \ref{principle:optimistic-cut} and \ref{principle:refinement}) would also be possible.
    The resulting cyclic-proof search strategy corresponds to the lazy abstraction in the style of \textsc{Blast}~\cite{Henzinger2002,Henzinger2004}.
    %
    \qed
\end{remark}

\section{Dual sequent and Dual \impact{}}
\label{sec:backward-symbolic-execution}

So far, we have discussed a relationship between cyclic-proof search for \( \iota(x) \vdash (\nu P)(x) \) and software model-checking algorithms.
As we \huchanged{emphasized} in Section~\ref{sec:idea}, the identification of the goal sequent \( \iota(x) \vdash (\nu P)(x) \) corresponding to software model-checking is an important contribution of this paper, and once the goal sequent is appropriately set, the contents of Sections~\ref{sec:base} and \ref{sec:IMPACT} are natural and expected, which one might find relatively straightforward.

This section briefly discusses proof-search for \( (\mu R)(x) \vdash \alpha(x) \) (where \( R \) is defined by \autoref{eq:reachable-states}) from the view point of software model-checking, and shows how the choice of the goal sequent affects the whole picture.

We derive the counterparts of \textsc{(SE)} and \textsc{(SE+Cut)}, called \textsc{(BackSE)} and \textsc{(BackSE+Cut)}.
We have
\begin{equation*}
    \infer[\textsc{($\mu$-L)}]{
        (\mu R)(x) \vdash \varphi(x)
    }{
        \iota(x) \vee (\exists y. (\mu R)(y) \wedge \tau(y,x)) \vdash \varphi(x)
    }
\end{equation*}
of which left-hand-side of \( \vdash \) in the premise is a complicated formula containing a least fixed-point.
Following the dual of Principle~\ref{principle:shape-of-goals}, moving logical connectives to the right-hand-side, we obtain the following derived rule:
\begin{equation*}
    \infer[\textsc{(\huchanged{BackSE})}.]{
        (\mu R)(x) \vdash \psi(x)
    }{
        (\mu R)(x) \vdash \forall y. \tau(x,y) \Rightarrow \psi(y)
        &&
        \iota(x) \models \psi(x)
    }
\end{equation*}
The formula \( \forall y. \tau(x,y) \Rightarrow \psi(y) \) has a clear computational interpretation: it is the weakest precondition of \( \psi \) with respect to \( \tau \).
This is the dual of the strongest \huchanged{postcondition}\sout{,} which coincides with the symbolic execution.

The duality becomes clearer in \textsc{(BackSE+Cut)}
\begin{equation*}
    \infer[\textsc{(BackSE+Cut)}]{
        (\mu R)(x) \vdash \psi(x)
    }{
        (\mu R)(x) \vdash \varphi(x)
        &&
        \varphi(x) \models \forall y. \tau(x,y) \Rightarrow \psi(y)
        &&
        \iota(x) \models \psi(x)
    }
\end{equation*}
which can be proved by
\begin{equation*}
    \infer[\textsc{(BackSE)}.]{
        (\mu R)(x) \vdash \psi(x)
    }{
        \infer[\textsc{(Cut)}]{
            \huchanged{(\mu R)(x) \vdash \forall y. \tau(x,y) \Rightarrow \psi(y)}
        }{
            (\mu R)(x) \vdash \varphi(x)
            &&
            \varphi(x) \models \forall y. \tau(x,y) \Rightarrow \psi(y)
        }
        &&
        \iota(x) \models \psi(x)
    }
\end{equation*}
Since \( \varphi(x) \models \forall y. \tau(x,y) \Rightarrow \psi(y) \) is equivalent to \( \exists x. \varphi(x) \wedge \tau(x,y)\models \psi(y) \), \textsc{(BackSE+Cut)} computes an \huchanged{over-approximation} of the \huchanged{symbolic} execution as in \textsc{(SE+Cut)}.
The only difference is the direction; in a partial proof for \( (\mu R)(x) \vdash \psi(x) \), an execution starts from a leaf and ends at the root, which should be understood as an assertion.

So model-checking algorithms corresponding to proof-search for \( (\mu R)(x) \vdash \psi(x) \) can be understood as procedures checking the backward (un)reachability.
Here is a twist: the natural logical expression \( (\mu R)(x) \vdash \psi(x) \) of the model-checking problem corresponds to the execution of programs in the unnatural direction.
We think that this twist might be a reason why the connection between cyclic proofs and software model-checking \sout{has}\huchanged{had} not yet \huchanged{been} formally revealed despite that it is natural and expected.


%% file: pdr.tex
\section{Maximal Conservativity and Property-Directed Reachability}
\label{sec:pdr}

Sections~\ref{sec:IMPACT} reviewed McMillan's approach based on interpolants~\cite{McMillan2006}, as well as its variants, from the view point of cyclic proof search.
His algorithm tentatively uses \( \top \) as the cut-formula at the beginning and improves it using interpolants.
The efficiency of the algorithm heavily relies on the quality of interpolants.

This section introduces a heuristic, or a criterion, that determines what is a ``better'' interpolation.
The criterion is \emph{conservativity}: it prefers an interpolation that introduces smaller changes to the current partial proof.
We study a variant of McMillan's algorithm that uses a \emph{maximally conservative} interpolation instead of arbitrary ones and discuss its connection to the approach to software model-checking called \emph{property directed reachability} (or \emph{PDR})~\cite{Hoder2012,Cimatti2012}.


\subsection{Maximal conservativity}
We introduce the key notion, \emph{maximal conservativity}, to the heuristic studied in this section.

Let us first recall McMillan's algorithm~\cite{McMillan2006}, regarded as a goal-oriented proof search.
The refinement process takes an invalid partial proof
\begin{equation*}
    \makebox[-70pt][0pt]{}
    \infer{
        \iota(x) \vdash (\nu P)(x)
    }{
        \infer{
            \varphi_1(x) \vdash (\nu P)(x)
        }{
            \infer*{
                \varphi_2(x) \vdash (\nu P)(x)
            }{
                \makebox[70pt][0pt]{}
                \infer{
                    \varphi_n(x) \vdash (\nu P)(x)
                    \makebox[70pt][0pt]{}
                }{
                    \top \vdash (\nu P)(x)
                    &&
                    \raisebox{1em}{\oalign{
                        \hfill $\varphi_n(x) \models \alpha(x)$
                        \crcr
                        $\exists x. \varphi_n(x) \wedge \tau(x,y) \models \top\;\;\;\;\;\!$
                    }}
                }
            }
            &&
            \hspace{-90pt}
            \raisebox{1em}{\oalign{
                \hfill $\varphi_1(x) \models \alpha(x)~\;\!$
                \crcr
                $\exists x. \varphi_1(x) \wedge \tau(x,y) \models \varphi_2(y)$
            }}
        }
        &&
        \hspace{-10pt}
        \raisebox{1em}{\oalign{
            \hfill $\iota(x) \models \alpha(x)~\;\!$
            \crcr
            $\exists x. \iota(x) \wedge \tau(x,y) \models \varphi_1(y)$
        }}
    }
\end{equation*}
in which all judgements except for \( \varphi_n(x) \models \alpha(x) \) are valid.
It returns a valid proof of \( \iota(x) \vdash (\nu P)(x) \) obtained by replacing each \( \varphi_i \) with a stronger predicate \( \varphi_i' \).

We represent the above partial proof as a sequence \( \varphi_0,\varphi_1,\dots,\varphi_n \) of formulas (where \( \varphi_0 = \iota \)).
Then an input to the refinement process can be represented as a sequence \( \varphi_0,\dots,\varphi_n \) of formulas that satisfies
\begin{equation}
    \varphi_0 = \iota
    \qquad\mbox{and}\qquad
    \varphi_i(x) \models \alpha(x)
    \mbox{ and }
    \varphi_i(x) \wedge \tau(x,y) \models \varphi_{i+1}(y)
    \mbox{ for } i = 0,\dots,n-1
    \label{eq:pdr:input-condition}
\end{equation}
(recall that all judgements in the input partial proof are valid except for \( \varphi_n(x) \models \alpha(x) \)).
The refinement process tries to construct a sequence \( \varphi'_0,\dots,\varphi'_n \) such that
\begin{equation}
    \begin{cases}
        \varphi'_0 = \iota \\
        \varphi'_i(x) \models \varphi_i(x) & \quad\mbox{for every \( i = 0,\dots,n \)} \\
        \varphi'_i(x) \wedge \tau(x,y) \models \varphi'_{i+1}(y) & \quad\mbox{for every \( i = 0,\dots,n-1 \)} \\
        \varphi'_n(x) \models \alpha(x).
    \end{cases}
    \label{eq:pdr:interpolation-constraints}
\end{equation}
We call a sequence that satisfies \autoref{eq:pdr:interpolation-constraints} a \emph{refinement} of \( \varphi_0,\dots,\varphi_n \).

In McMillan's approach~\cite{McMillan2006}, we let \( \varphi_i'(x) \equiv \varphi_i(x) \wedge Q_i(x) \) (\( i = 0,\dots,n \)), where \( Q_i \) is a fresh predicate variable, and solve the constraints
\begin{equation}
    \{~~ (\varphi_i(x) \wedge Q_i(x)) \wedge \tau(x,y) \Rightarrow Q_{i+1}(y) \mid i = 0,\dots,n-1 ~~\} \cup \{~ Q_0 = \top,~ \varphi_n(x) \wedge Q_n(x) \Rightarrow \alpha(x) ~\}
    \label{eq:pdr:constraints-for-refinement}
\end{equation}
using an interpolating theorem prover, finding a suitable assignment for \( (Q_i)_{0 \le i \le n} \).
The constraint solving of (\ref{eq:pdr:constraints-for-refinement}) is a kind of interpolation problem, known as \emph{sequential interpolation}~\cite{Henzinger2004,McMillan2006}, we call a solution \( (Q_i)_{0 \le i \le n} \) an \emph{interpolant}.
Observe that this refinement process may change any formula in the current partial proof.

The heuristic studied in this section tries to minimize the changes.
That means, it prefers an interpolant \( (Q_i)_{0 \le i \le n} \) of (\ref{eq:pdr:constraints-for-refinement}) such that \( Q_i = \top \) for more indices \( i \).
It is not difficult to see that, if \( Q_0,\dots,Q_n \) is an interpolant and \( Q_k = \top \), then \( \top,\dots,\top,Q_{k+1},\dots,Q_n \) is also an interpolant.
Hence the heuristic chooses an interpolant \( \top,\dots,\top,Q_k,Q_{k+1},\dots,Q_n \) such that no interpolant of the form \( \top,\dots,\top,\top,Q'_{k+1},\dots,Q'_n \) exists, and we call such an interpolant a \emph{maximally conservative interpolant}.
The corresponding \emph{maximally conservative refinement} \( \varphi_0',\dots,\varphi'_n \), which is defined by \( \varphi'_i = \varphi_i \wedge Q_i \) (\( 0 \le i \le n \)), satisfies \( \varphi'_i = \varphi_i \) as much as possible.
\begin{definition}[Maximally conservative refinement]
    Let \( \varphi_0,\dots,\varphi_n \) be a sequential representation of a partial proof that satisfies (\ref{eq:pdr:input-condition}).
    Its \emph{refinement} is a valid partial proof represented as \( \varphi'_0,\dots,\varphi'_n \) that satisfies the constraints in (\ref{eq:pdr:interpolation-constraints}).
    It is \emph{\( k \)-conservative} if \( \varphi'_j = \varphi_j \) for every \( j \le k \).
    It is \emph{maximally conservative} if it is \( k \)-conservative and no \( (k+1) \)-conservative refinement exists.
    \qed
\end{definition}


The heuristic based on maximal conservativity introduces the following principle.
\begin{principle}\label{principle:maximally-conservative}
    Use a maximally-conservative-interpolation theorem prover as the backend solver.\footnote{
        The actual implementation of \impact{} follows this principle, as we mentioned in Remark~\ref{rem:impact-diff}.
    }
\end{principle}

We define \emph{lazy abstraction with maximally conservative interpolants}, or \MCImpact{}, as a variant of McMillan's algorithm in which the interpolating theorem prover always chooses a maximally conservative one.
So \textsc{Impact/mc} is an algorithm following Principles~\ref{principle:shape-of-goals}, \ref{principle:optimistic-cut}, \ref{principle:refinement} and \ref{principle:maximally-conservative}.
Since Principle~\ref{principle:maximally-conservative} does not specify an algorithm that computes a maximally conservative interpolant, the behavior of \textsc{Impact/mc} depends on the backend maximally-conservative-interpolation theorem prover.

The rest of this section develops algorithms for maximally conservative interpolation and compares \textsc{Impact/mc} with other software model-checking algorithms.
For technical convenience, we shall focus on maximally conservative refinement instead of interpolation in the sequel.

\begin{remark}
    The idea of maximally conservative refinement or interpolant can be found in the literature, e.g.~in \cite{Cimatti2012} and \cite{Vizel2014}.
    We shall compare their work with this paper in Remark~\ref{rem:originality-of-maximal-conservativity}, after establishing a connection between maximal conservativity and PDR.
    \qed
\end{remark}

\subsection{A Na\"ive algorithm for maximally conservative refinement}
A maximally conservative refinement is computable if binary interpolants (and thus sequential interpolants) are computable.
Here we give a na\"ive algorithm in order to prove the computability; more sophisticated algorithms shall be introduced in the following subsections.

Algorithm~\ref{alg:simple-mcr} is a na\"ive algorithm that constructs a maximally conservative refinement.
The na\"ive algorithm tries to construct a \( k \)-conservative refinement for \( k = n, n-1, n-2,\dots,0 \) until it finds a refinement.
A \( k \)-conservative refinement can be computed using a sequential interpolation algorithm: letting \( \varphi'_i(x) = \varphi_i(x) \wedge Q_i(x) \), where \( Q_i \) is a fresh predicate variable, it suffices to solve the constraints (\ref{eq:pdr:constraints-for-refinement}) with \( \top \Rightarrow Q_i(x) \) for every \( i \le k \).

\begin{algorithm}[t]
    \caption{~~Na\"ive algorithm to compute a maximally conservative refinement}\label{alg:simple-mcr}
    \begin{algorithmic}[1]
    \Require{Predicates \( \iota(x), \tau(x,y), \alpha(x) \) defining linear CHCs}
    \Require{A partial proof $ \varphi_0,\dots,\varphi_n $ satisfying (\ref{eq:pdr:input-condition})}
    \Ensure{A maximally conservative refinement $ \varphi_0',\dots,\varphi_n'$}
    \Comment{\textsc{Na\"iveMCR} may fail}
    \Function{Na\"iveMCR}{$\iota,\tau,\alpha;~~ \varphi_0,\dots,\varphi_n$} 
    \State $ \mathcal{C} := \{~~ (\varphi_i(x) \wedge Q_{i}(x)) \wedge \tau(x,y) \Rightarrow Q_{i+1}(y) ~\mid i = 0,\dots,n-1 \,\} \cup \{~ \varphi_n(x) \wedge Q_n(x) \Rightarrow \alpha(x) ~\} $
    \State $ j := n $
    \While{$ j \geq 0 $}
        \If {$\mathcal{C} \cup \{ \top \Rightarrow Q_i(x) \mid i = 0,\dots,j \}$ has a solution}
            \State \textbf{let} $Q_0,\dots,Q_n$ \textbf{be} a solution
            \State \textbf{return} $(\varphi_i \wedge Q_i)_{i=0,\dots,n}$
        \EndIf
        \State $ j := j - 1 $
    \EndWhile
    \State $\mathbf{fail}$
    \EndFunction
    \end{algorithmic}
\end{algorithm}

\subsection{Inductive characterization of maximally conservative refinement}
We show that a subproof of a maximally conservative refinement must be maximally conservative as well in a certain sense, and this property gives us an inductive characterization of maximally conservative refinements.
After giving an inductive characterization of maximally conservative refinements, we develop an algorithm inspired by the characterization.
We say that \( \varphi'_0,\dots,\varphi'_n \) is a \emph{maximally conservative refinement of \( \varphi_0,\dots,\varphi_n \) satisfying \( \psi(x) \)} if it is a maximally conservative refinement with respect to the constraint system\footnote{The class of constraint systems of this form is called \emph{linear CHCs} (see e.g., \cite{Bjorner2015a} for the definition).} \( S' = \{ \iota(x) \Rightarrow P(x), P(x) \wedge \tau(x,y) \Rightarrow P(y), P(x) \Rightarrow \psi(x) \} \), in which the assertion \( \alpha \) is replaced with \( \psi \).

\begin{lemma}\label{lem:pdr:subrefine}
    Let \( (\varphi_i)_{0\le i \le n} \) be a partial proof satisfying (\ref{eq:pdr:input-condition}).
    For every partial proof \( (\psi_0,\dots,\psi_{n-1}) \), the following conditions are equivalent:
    \begin{enumerate}
        \renewcommand{\labelenumi}{(\roman{enumi})}
        \item There exists \( \psi_n \) such that \( (\psi_0,\dots,\psi_{n-1},\psi_n) \) is a refinement of \( (\varphi_0,\dots,\varphi_n) \) satisfying \( \alpha \).
        \item \( (\psi_0,\dots,\psi_{n-1}) \) is a refinement of \( (\varphi_0,\dots,\varphi_{n-1}) \) satisfying \( \neg \exists y. \tau(x,y) \wedge \neg \alpha(y) \).
    \end{enumerate}
\end{lemma}
\begin{proof}
    Assume (1).
    Then \( \psi_{n-1}(x) \wedge \tau(x,y) \models \psi_n(x) \) and \( \psi_n(x) \models \alpha(x) \).  Hence \( \psi_{n-1}(x) \wedge \tau(x,y) \models \alpha(y) \), which implies \( \psi_{n-1}(x) \models \neg \exists y. \tau(x,y) \wedge \neg \alpha(y) \).

    Assume (2).
    Then \( \psi_{n-1}(x) \models \neg \exists y. \tau(x,y) \wedge \neg \alpha(y) \), which implies \( \psi_{n-1}(x) \wedge \tau(x,y) \models \alpha(y) \).
    Since \( \psi_{n-1}(x) \models \varphi_{n-1}(x) \) and \( \varphi_{n-1}(x) \wedge \tau(x,y) \models \varphi_{n}(y) \) by the assumption (\ref{eq:pdr:input-condition}), \( \psi_{n-1}(x) \wedge \tau(x,y) \models \varphi_n(y) \) as well.
    Letting \( \psi_n(x) \) be an interpolant of \( \psi_{n-1}(x) \wedge \tau(x,y) \) and \( \alpha(y) \wedge \varphi_n(y) \), the partial proof \( (\psi_0,\dots,\psi_{n-1},\psi_n) \) is a refinement of \( (\varphi_0,\dots,\varphi_n) \) satisfying \( \alpha \).
\end{proof}
\begin{theorem}\label{thm:mcr-ind}
    Let \( (\varphi_i)_{0\le i \le n} \) be a partial proof satisfying (\ref{eq:pdr:input-condition}).
    Suppose that \( (\varphi'_i)_{0 \le i \le n} \) is a refinement of \( (\varphi_i)_{0 \le i \le n} \) satisfying \( \alpha \).
    It is maximally conservative if and only if the following conditions hold:
    \begin{enumerate}
        \renewcommand{\labelenumi}{(\alph{enumi})}
        \item If \( \varphi_n(x) \models \alpha(x) \), then \( (\varphi'_0,\dots,\varphi'_n) = (\varphi_0,\dots,\varphi_n) \).
        \item If \( \varphi_n(x) \not\models \alpha(x) \), then \( (\varphi'_0,\dots,\varphi'_{n-1}) \) is a maximally conservative refinement of \( (\varphi_0,\dots,\varphi_{n-1}) \) satisfying \( \neg \exists y. \tau(x,y) \wedge \neg\alpha(y) \).
    \end{enumerate}
\end{theorem}
\begin{proof}
    The first case is trivial.  Assume that \( \varphi_n(x) \not\models \alpha(x) \).  Then there is no \( n \)-conservative refinement.
    So it suffices to show that the following conditions are equivalent for every \( k < n \):
    \begin{itemize}
        \item \( (\varphi_i)_{0 \le i \le n} \) has a \( k \)-conservative refinement satisfying \( \alpha \).
        \item \( (\varphi_i)_{0 \le i \le n-1} \) has a \( k \)-conservative refinement satisfying \( \neg \exists y. \tau(x,y) \wedge \neg \alpha(y) \).
    \end{itemize}
    This equivalence is a consequence of Lemma~\ref{lem:pdr:subrefine}.
\end{proof}

\IndMCR{} (Algorithm~\ref{alg:rec-mcr}) is an algorithm based on the characterization in Theorem~\ref{thm:mcr-ind}.
The function \( \mathit{interpolation}(\psi(\vec{x},\vec{y}), \vartheta(\vec{y},\vec{z})) \) nondeterministically returns a formula from \( \{ \gamma(\vec{y}) \mid \psi(\vec{x},\vec{y}) \models \gamma \mbox{ and } \gamma(\vec{y}) \models \vartheta(\vec{y},\vec{z}) \} \), provided that \( \psi(\vec{x},\vec{y}) \models \vartheta(\vec{y},\vec{z}) \).
Line 3 corresponds to the case (a), whereas lines 4, 5 and 7 correspond to the case (b).
Hence the correctness of \IndMCR{} is a direct consequence of Theorem~\ref{thm:mcr-ind}.
The termination can be shown by induction on \( n \), since it has no loop and a recursive call takes the subproof \( (\varphi_0,\dots,\varphi_{n-1}) \) as its argument; \tkchanged{note that \( \varphi_0 = \iota \) and thus, if \( n=0 \), exactly one of the conditions in lines 2 and 3 is true}.
In summary, we have the following theorem.
\begin{theorem}
    \IndMCR{} always terminates.
    The set of possible return values of \IndMCR{} coincides with the set of maximally conservative refinements.
    \qed
\end{theorem}

\begin{algorithm}[t]
    \caption{~~Maximally conservative refinement, defined inductively}\label{alg:rec-mcr}
    \begin{algorithmic}[1]
    \Require{Predicates \( \iota(x), \tau(x,y), \alpha(x) \) defining linear CHCs}
    \Require{A partial proof $ \varphi_0,\dots,\varphi_n $ satisfying (\ref{eq:pdr:input-condition})}
    \Ensure{A maximally conservative refinement $ \varphi_0',\dots,\varphi_n'$} \Comment{\textsc{IndMCR} may fail}
    \Function{IndMCR}{$\iota,\tau,\alpha;~~ \varphi_0,\dots,\varphi_n$}
    \State {\textbf{if} \( \models \exists x. \iota(x) \wedge \neg \alpha(x) \) \textbf{then} \textbf{fail}}
    \State {\textbf{if} \( \varphi_n(x) \models \alpha(x) \) \textbf{then} \textbf{return} $(\varphi_0,\dots,\varphi_n)$}
    \State $ \gamma(x) := \exists y. \tau(x,y) \wedge \neg \alpha(y) $
    \State $ (\varphi_0,\dots,\varphi_{n-1}) := \textsc{IndMCR}(\iota,\tau,\neg \gamma;~~ \varphi_0,\dots,\varphi_{n-1})$
    \\ \Comment{\( \exists x. \varphi_{n-1}(x) \wedge \tau(x,y) \models \varphi_n(y) \wedge \alpha(y) \) holds}
    \State $\varphi_n(y) := \mathit{interpolation}\big(\varphi_{n-1}(x) \wedge \tau(x,y),~~~~ \varphi_n(y) \wedge \alpha(y) \big)$
    \Comment{Nondeterministic}
    \State \textbf{return} $(\varphi_0,\dots,\varphi_n)$
    \EndFunction
    \end{algorithmic}
\end{algorithm}

\subsection{Approximating quantified-formulas by quantifier-free formulas}
\IndMCR{} given in the previous subsection is defined by a simple induction on the structure of the partial proof.
However it is not practical since it uses quantified formulas, which are hard to reason about.
This section introduces an algorithm that \tkchanged{does not pass a quantified formula as an argument}.

Our goal is to remove the quantifier of \( \exists y. \tau(x,y) \wedge \neg \alpha(y) \) in \IndMCR{}.
We appeal to the logical equivalence
\begin{equation}
    \exists y. \tau(x,y) \wedge \neg \alpha(y)
    \qquad\Leftrightarrow\qquad
    \bigvee_{c} (\tau(x,c) \wedge \neg \alpha(c))
\end{equation}
where the right-hand-side is the infinite disjunction indexed by concrete values \( c \).
This equivalence allows us to decompose a single complicated (i.e.~quantified) formula into infinitely-many simple (i.e.~quantifier-free) formulas.
We reduce the computation of a maximally conservative refinement satisfying \( \neg(\exists y. \tau(x,y) \wedge \neg \alpha(y)) \) to iterative computation of maximally conservative refinements satisfying \( \neg(\tau(x,c_i) \wedge \neg \alpha(c_i)) \) for \( i = 1,2,3,\dots \), where \( \{ c_1,c_2,\dots \} \) is an enumeration of values.

\IndPDR{} (Algorithm~\ref{alg:simple-pdr}) is an algorithm based on this idea.
As its name suggests, it is closely related to a family of algorithms known as \emph{property-directed reachability}.
We shall discuss the connection in the next subsection.

\begin{algorithm}[t]
    \caption{~~PDR defined by structural induction on partial proofs}\label{alg:simple-pdr}
    \begin{algorithmic}[1]
    \Require{Predicates \( \iota(x), \tau(x,y), \alpha(x) \) defining linear CHCs}
    \Require{A partial proof $ \varphi_0,\dots,\varphi_n $ satisfying (\ref{eq:pdr:input-condition})}
    \Ensure{A maximally conservative refinement $ \varphi_0',\dots,\varphi_n'$} \Comment{\textsc{IndPDR} may fail or diverge}
    \Function{\IndPDR}{$\iota,\tau,\alpha;~~ \varphi_0,\dots,\varphi_n$}
    \State {\textbf{if} \( \models \exists x. \iota(x) \wedge \neg \alpha(x) \) \textbf{then} \textbf{fail}}
    \State {\textbf{if} \( \varphi_n(x) \models \alpha(x) \) \textbf{then} \textbf{return} $(\varphi_0,\dots,\varphi_n)$}
    \While{$ \models \exists x. \exists y. \varphi_{n-1}(x) \wedge \tau(x,y) \wedge \neg \alpha(y)$}
        \State \textbf{let} $M$ \textbf{be} an assignment such that $ M \models \varphi_{n-1}(x) \wedge \tau(x,y) \wedge \neg \alpha(y) $
        \State $ \gamma(x) := \tau(x,M(y)) \wedge \neg \alpha(M(y)) $
        \State $ (\varphi_0,\dots,\varphi_{n-1}) := \IndPDR(\iota,\tau,\neg\gamma;~~ \varphi_0,\dots,\varphi_{n-1})$
    \EndWhile
    \State{} \Comment{\( \exists x. \varphi_{n-1}(x) \wedge \tau(x,y) \models \varphi_n(y) \wedge \alpha(y) \) holds}
    \State $\varphi_n(y) := \mathit{interpolation}\big(\varphi_{n-1}(x) \wedge \tau(x,y),~~~~ \varphi_n(y) \wedge \alpha(y) \big)$
    \Comment{Nondeterministic}
    \State \textbf{return} $(\varphi_0,\dots,\varphi_n)$
    \EndFunction
    \end{algorithmic}
\end{algorithm}

The loop at lines~4--8 is the computation of a maximally conservative refinement satisfying \( \neg \exists y. \tau(x,y) \wedge \neg \alpha(y) \).
If the algorithm actually computed maximally conservative refinement satisfying \( \tau(x,c) \wedge \neg \alpha(c) \) for all \( c \), it would never terminate.
So \IndPDR{} computes it only for relevant values.
Line~5 chooses a value \( M(y) \) and line~7 computes a maximally conservative refinement satisfying \( \neg (\tau(x,c) \wedge \neg \alpha(c)) \) where \( c = M(y) \).
Note that this is an over-approximation:
\begin{equation*}
    \neg \exists y. \tau(x,y) \wedge \neg \alpha(y) \models \neg (\tau(x,c) \wedge \neg \alpha(c)).
\end{equation*}
If the loop terminates, then \( \models \neg \exists x. \exists y. \varphi_{n-1}(x) \wedge \tau(x,y) \wedge \neg \alpha(y) \), which implies \( \varphi_{n-1}(x) \models \neg \exists y. \tau(x,y) \wedge \neg \alpha(y) \) as expected.

We prove the correctness.
\begin{lemma}\label{lem:pdr:iterative-mcr}
    Let \( (\varphi_i)_{0 \le i \le n} \) be a partial proof and let \( \alpha_1, \alpha_2 \) be formulas such that \( \alpha_2(x) \models \alpha_1(x) \).
    Suppose that \( (\varphi'_i)_{0 \le i \le n} \) is a maximally conservative refinement of \( (\varphi_i)_i \) satisfying \( \alpha_1 \) and that \( (\varphi''_i)_{0 \le i \le n} \) is a maximally conservative refinement of \( (\varphi'_i)_{0 \le i \le n} \) satisfying \( \alpha_2 \).
    Then \( (\varphi''_i)_i \) is a maximally conservative refinement of \( (\varphi_i)_i \) satisfying \( \alpha_2 \).
\end{lemma}
\begin{proof}
    Assume a \( k \)-conservative refinement \( (\psi_i)_{0 \le i \le n} \) of \( (\varphi_i)_i \) satisfying \( \alpha_2 \).
    It suffices to show that \( (\varphi''_i)_i \) is also \( k \)-conservative.
    Since \( \alpha_2(x) \models \alpha_1(x) \), the partial proof \( (\psi_i)_i \) is also a refinement of \( (\varphi_i)_i \) satisfying \( \alpha_1 \).
    Since \( (\varphi'_i)_i \) is maximally conservative, it is \( k \)-conservative.
    Then \( (\psi_i \wedge \varphi'_i)_{0 \le i \le n} \) is a \( k \)-conservative refinement of \( (\varphi'_i)_i \) satisfying \( \alpha_1 \).
    By the maximal conservativity of \( (\varphi''_i)_i \), it is \( k \)-conservative as well.
\end{proof}

\begin{lemma}
    If \IndPDR{} terminates, it returns a maximally conservative refinement.
\end{lemma}
\begin{proof}
    Let \( (\psi_i)_{0 \le i \le n} \) be the initial value of \( (\varphi_i)_{0 \le i \le n} \) and \( (\psi_i')_{0 \le i \le n} \) be the value of \( (\varphi_i)_i \) at the end of the loop (or, equivalently, line~9).
    It suffices to show that \( (\psi'_i)_{0 \le i \le n-1} \) is a maximally conservative refinement of \( (\psi_i)_{0 \le i \le n-1} \) satisfying \( \neg \exists y. \tau(x,y) \wedge \neg \alpha(y) \).
    Clearly \( \psi'_{n-1}(x) \models \neg \exists y. \tau(x,y) \wedge \neg \alpha(y) \) since the condition in line~4 is false for the values at line~9.
    The maximal conservativity \tkchanged{can be proved by induction on the number of iterations}
    using Lemma~\ref{lem:pdr:iterative-mcr} since \( \neg \exists y. \tau(x,y) \wedge \neg \alpha(y) \models \neg (\tau(x,c) \wedge \neg \alpha(c)) \) for every \( c \).
\end{proof}

Other properties of \IndPDR{} shall be discussed below.

\subsection{PDR vs.~\IndPDR{} as a transition system}
\emph{Property-directed reachability} (or \emph{PDR}) refers to the IC3 algorithm~\cite{Bradley2011,Een2011} and its derivatives~\cite{Hoder2012,Cimatti2012,Cimatti2014,Komuravelli2013,Komuravelli2014,Birgmeier2014,Beyer2020}.
This section compares PDR and \MCImpact{}, in particular \MCImpact{} with \IndPDR{}.
Since PDR is formalized as an abstract transition system in \cite{Hoder2012}, we give a transition system corresponding to \MCImpact{} and compare it with the transition system in \cite{Hoder2012}.

\paragraph{\IndPDR{} as a transition system}
An abstract transition system for \IndPDR{} can be derived by the following argument.
The starting point is operational semantics of the program \IndPDR{}.
Since \IndPDR{} is a recursive program, a configuration in the operational semantics consists of the current assignment to variables and the call-stacks.
So a configuration can be written as
\begin{equation*}
    (\alpha';~~\varphi_0,\dots,\varphi_{k-1}) \blacksquare
    (M^{k}, \alpha^k;~~\varphi^k_0,\dots,\varphi^k_k) (M^{k+1},\alpha^{k+1};~~\varphi^{k+1}_0,\dots,\varphi^{k+1}_{k+1}) \dots (M^n,\alpha^n;~~\varphi^n_0,\dots,\varphi^n_n),
\end{equation*}
where the left-hand-side of \( \blacksquare \) is the current variable assignment and the right-hand-side is the call-stack.
An abstract transition system can be obtained by using this representation of the call-stack.

This representation, however, has redundancy.
First forgetting \( \varphi^\ell_0,\dots,\varphi^\ell_{\ell-1} \) in the frame \( (M^{\ell}, \alpha^\ell;~~\varphi^\ell_0,\dots,\varphi^\ell_\ell) \) does not cause any problem, since they shall be overwritten when the control will come back to the frame.
Second \( \alpha^\ell \) is redundant: \( \alpha^n = \alpha \) and \( \alpha^\ell(x) = (\tau(x,M^{\ell+1}(y)) \wedge \neg \alpha^{\ell+1}(M^{\ell+1}(y))) \) for \( \ell = n-1,\dots,k \).
Third \( \alpha' = (\tau(x,M^{k}(y)) \wedge \neg \alpha^{k}(M^k(y))) \).
Fourth \( M^\ell(x) \) is not used.
So we obtain a compact representation:
\begin{equation*}
    (M^k(y) \dots M^n(y) \parallel \varphi_0 \dots \varphi_{k-1} \varphi^k_k \varphi^{k+1}_{k+1} \dots \varphi^n_n).
\end{equation*}
Figure~\ref{fig:SPDR} defines the abstract transition system corresponding to \IndPDR{}, which is the operational semantics of \IndPDR{} with the above representation of configurations.

\begin{figure}[t]
    \begin{gather*}
        \infer[\textsc{(Candidate)}]{
            (\epsilon \parallel \varphi_0 \dots \varphi_n) \longrightarrow (c \parallel \varphi_0 \dots \varphi_n)
        }{
            \models \varphi_n(c) \wedge \neg \alpha(c)
        }
        \\[5pt]
        \infer[\textsc{(Decide)}]{
            (c_k\,\dots\,c_n \parallel \varphi_0 \dots \varphi_n) \longrightarrow (c' c_k \dots c_n \parallel \varphi_0 \dots \varphi_n)
        }{
            \models \neg \iota(c_k)
            \qquad
            \models \varphi_{k-1}(c') \wedge \tau(c', c_k)
        }
        \\[5pt]
        \infer[\textsc{(Conflict)}]{
            (c_k \dots c_n \parallel \varphi_0 \dots \varphi_n)
            \longrightarrow
            (c_{k+1} \dots c_n \parallel \varphi_0 \dots \varphi_{k-1} (\varphi_{k} \wedge \psi) \varphi_{k+1} \dots \varphi_n)
        }{
            \varphi_{k-1}(x) \wedge \tau(x,y) \models \psi(y)
            \qquad
            \psi(x) \models x \neq c_k
        }
    \end{gather*}
    \vspace{-20pt}
    \caption{The transition rules for \IndPDR{}.}
    \label{fig:SPDR}
\end{figure}

The above argument shows the following proposition.
\begin{proposition}
    \( \IndPDR{}(\iota,\tau,\alpha;~~\varphi_0,\dots,\varphi_n) \) may return \( (\varphi'_0,\dots,\varphi'_n) \) if and only if \( (\epsilon \parallel \varphi_0,\dots,\varphi_n) \longrightarrow^* (\epsilon \parallel \varphi'_0,\dots,\varphi'_n) \) with \( \varphi'_n(x) \models \alpha(x) \).
    \qed
\end{proposition}

Before comparing the transition system with PDR, we prove a part of the correctness of \IndPDR{} using the transition system.
\begin{lemma}
    If \( (\varphi'_i)_{0 \le i \le n} \) is a maximally conservative refinement of \( (\varphi_i)_{0 \le i \le n} \), then
    \begin{equation*}
        (\epsilon \parallel \varphi_0 \dots \varphi_n) \longrightarrow^* (\epsilon \parallel \varphi'_0 \dots \varphi'_n).
    \end{equation*}
\end{lemma}
\begin{proof}
    Assume that \( (\varphi'_i)_{0 \le i \le n} \) is a maximally conservative refinement of \( (\varphi_i)_{0 \le i \le n} \) satisfying \( \alpha \).
    Let \( k \) be the natural number such that \( \varphi'_k = \varphi_k \) but \( \varphi'_{k+1} \neq \varphi_{k+1} \).
    For \( i = k+2,\dots,n \), we define \( \psi_i \) by
    \begin{equation*}
        \varphi_{k+1} \rightsquigarrow
        \psi_{k+2} \rightsquigarrow
        \psi_{k+3} \rightsquigarrow
        \dots \rightsquigarrow
        \psi_n
    \end{equation*}
    where \( \rightsquigarrow \) is the precise symbolic computation.
    Then \( (\varphi_0,\dots,\varphi_{k+1},\psi_{k+2},\dots,\psi_n) \) satisfies the requirements for refinement except for \( \psi_n(x) \models \alpha(x) \).
    By the maximal conservativity of \( (\varphi'_i)_i \), we know that \( (\varphi_0,\dots,\varphi_{k+1},\psi_{k+2},\dots,\psi_n) \) is not a refinement satisfying \( \alpha \), and thus \( \models \psi_n(c_n) \wedge \neg \alpha(c_n) \) for some \( c_n \).
    Hence we have a sequence \( (c_{k+1}, \dots, c_n) \) of values such that \( \models \varphi_{k+1}(c_{k+1}) \) and \( \models \tau(c_i,c_{i+1}) \) for every \( i = k+1,\dots,n-1 \).
    So we have the following transition sequence:
    \begin{align*}
        &
        (\epsilon \parallel \varphi_0\dots\varphi_n)
        \longrightarrow^*
        (c_{k+1} \dots c_n \parallel \varphi_0 \dots \varphi_n)
        \longrightarrow
        (c_{k+2} \dots c_n \parallel \varphi_0 \dots \varphi_k \varphi'_{k+1} \varphi_{k+2} \dots \varphi_n)
        \\
        &
        \longrightarrow
        (c_{k+3} \dots c_n \parallel \varphi_0 \dots \varphi_k \varphi'_{k+1} \varphi'_{k+2} \varphi_{k+3}\dots \varphi_n)
        \longrightarrow
        (\epsilon \parallel \varphi_0 \dots \varphi_k \varphi'_{k+1} \varphi'_{k+2} \cdots \varphi'_n).
    \end{align*}
\end{proof}

\begin{corollary}
    The set of possible return values of \IndPDR{} coincides with the set of all maximally conservative refinements.
\end{corollary}

\paragraph{Definition of PDR}
We review PDR following Hoder and Bj{\o}rner~\cite{Hoder2012}.
\tkchanged{Here we introduce a simplified version, which we call \emph{simplified PDR} (\emph{sPDR}); see Remark~\ref{rem:pdr:practical} for the differences}.
A \emph{configuration} is of the form \( (c_k \dots c_n \parallel F_0 F_1 \dots F_n) \), consisting of two sequences of \emph{(candidate) counterexamples} \( c_i \in \mathcal{S} \) and \emph{frames} \( F_i \),%
    \footnote{We do not use metavariables \( \varphi \) and \( \psi \) for frames, because a frame does not precisely correspond to a cut-formula.  See a discussion below.}
which are formulas denoting unary predicates.
We call \( c_i \) (resp.~\( F_i \)) the counterexample (resp.~the frame) at \emph{level} \( i \).
Note that counterexamples are associated to only higher levels and that the counterexample sequence can be empty \( \epsilon \).
Intuitively a frame \( F_i \) at level \( i \) denotes an over-approximation of states of distance \( \le i \) and \( c_i \) together with \( c_{i+1} \dots c_n \) is a witness of unsafety of \( F_i \) with bound \( (n-i) \), i.e.~\( F_i(c_i) \) and \( \bigwedge_{j=i}^{n-1} \tau(c_j, c_{j+1}) \wedge \neg \alpha(c_n) \).
Formally each configuration \( (c_k \dots c_n \parallel F_0 \dots F_n) \) must satisfy (a)~\( F_0 = \iota \), (b)~\( F_i(x) \vee \mathcal{F}(F_i)(x) \models F_{i+1}(x) \) for every \( 0 \le i \le n \),%
    \footnote{\( \mathcal{F} \) is the predicate transformer given by \( \mathcal{F}(F)(y) := \iota(y) \vee (\exists x. F(x) \wedge \tau(x,y)) \).}
(c) \( F_{i}(x) \models \alpha(x) \) for every \( 0 \le i < n \), (d)~\( \models F_i(c_i) \) for every \( k \le i \le n \), and (e)~\( \models \bigwedge_{i=k}^{n-1} \tau(c_i, c_{i+1}) \wedge \neg \alpha(c_n) \).


The transition rules are given in \autoref{fig:PDR}.
The rule \textsc{(Unfold)} is applicable when \( F_0 = \iota \) has been confirmed safe with bound \( n \).
This rule starts to check if \( \iota \) is safe with bound \( (n+1) \), by extending the frame by \( \top \).
\textsc{(Candidate)} is applicable when a state \( c_n \) satisfying \( F_n \) violates the assertion \( \alpha \) and starts to check if the bad state \( c_n \) is \( n \)-step reachable from \( \iota \).
The candidate counterexample is propagated backward by \textsc{(Decide)} rule.
If \textsc{(Decide)} is not applicable, the candidate is spurious, i.e.~it was introduced by the approximation at this level.
Then \textsc{(Conflict)} refines the approximation, strengthening \( F_k \) so that it does not contain the spurious counterexample \( c_k \) (i.e.~\( \psi(x) \models x \neq c_k \)).
\begin{figure}[t]
    \begin{gather*}
        \infer[\textsc{(Candidate/sPDR)}]{
            (\epsilon \parallel F_0 \dots F_n) \longrightarrow_{\textsc{pdr}} (c \parallel F_0 \dots F_n)
        }{
            \models F_n(c) \wedge \neg \alpha(c)
        }
        \\[5pt]
        \infer[\textsc{(Decide/sPDR)}]{
            (c_k\,\dots\,c_n \parallel F_0 \dots F_n) \longrightarrow_{\textsc{pdr}} (c' c_k \dots c_n \parallel F_0 \dots F_n)
        }{
            \models F_{k-1}(c') \wedge \tau(c', c_k)
        }
        \\[5pt]
        \infer[\textsc{(Conflict/sPDR)}]{
            (c_k \dots c_n \parallel F_0 \dots F_n)
            \longrightarrow_{\textsc{pdr}}
            (c_{k+1} \dots c_n \parallel F_0\dots F_{k-1} (F_{k} \wedge \psi) F_{k+1} \dots F_n)    
        }{
            \iota(y) \vee \big(\exists x. F_{k-1}(x) \wedge \tau(x,y))\big) \models \psi(y)
            \qquad
            F_{k-1}(x) \models \psi(x)
            \qquad
            \psi(x) \models x \neq c_k
        }
    \end{gather*}
    \vspace{-20pt}
    \caption{A \tkchanged{simplified version} of the transition rules for PDR in \cite{Hoder2012}; see Remark~\ref{rem:pdr:practical} for the differences.  We omit \textsc{(Unfold)}, \textsc{(Initialize)}, \textsc{(Valid)} and \textsc{(Model)}, as these rules are not a part of the refinement/interpolation process.  These rules corresponds the content described in Definition~\ref{def:pdr-whole-process}.}
    \label{fig:PDR}
\end{figure}

\begin{definition}[Simplified PDR]\label{def:pdr-whole-process}
    The algorithm sPDR is defined as follows.
    It starts from the \emph{initial configuration} \( (\epsilon \parallel \iota) \).
    The algorithm iteratively rewrites the configuration.
    If it reaches \( (\epsilon \parallel F_0 \dots F_n) \) with \( F_\ell(x) \models F_{\ell-1}(x) \) for some \( \ell \), then the program is safe and an invariant is \( F_\ell \).
    If it reaches \( (c_0 \dots c_n \parallel F_0 \dots F_n) \) with a counterexample at level \( 0 \), the program is unsafe.
    If it reaches \( (\epsilon \parallel F_0 \dots F_n) \) with \( F_n(x) \models \alpha \), then it moves to \( \epsilon \parallel F_0 \dots F_n \top) \).
    \qed
\end{definition}

\paragraph{sPDR vs.~\MCImpact{} with \IndPDR{}}
sPDR in \autoref{fig:PDR} is essentially the same as \MCImpact{} with \IndPDR{}.
Consider the relation \( \approx \) between configurations of \IndPDR{} and of sPDR defined by
\begin{equation*}
    \textstyle
    (c_k \dots c_n \parallel \varphi_0 \dots \varphi_n)
    \approx
    (c_k \dots c_n \parallel F_0 \dots F_n)
    \quad\mbox{iff}\quad
    \models F_j(x) \Leftrightarrow \bigvee_{0 \le i \le j} \varphi_i(x)
    \mbox{ for every \( 0 \le j \le n \)}
\end{equation*}
This relation \( \approx \) is a bisimilation.
\begin{theorem}\label{thm:PDR-MC}
    Suppose that \( (\vec{c} \parallel \vec{\varphi}) \approx (\vec{c} \parallel \vec{F}) \).
    If \( (\vec{c} \parallel \vec{\varphi}) \longrightarrow (\vec{c}' \parallel \vec{\varphi}') \), then \( (\vec{c} \parallel \vec{F}) \longrightarrow_{\textsc{pdr}} (\vec{c}' \parallel \vec{F}') \) and \( (\vec{c}' \parallel \vec{\varphi}') \approx (\vec{c}', \parallel \vec{F}') \) for some \((\vec{c}' \parallel \vec{F}')\).
    Conversely, if \( (\vec{c} \parallel \vec{F}) \longrightarrow_{\textsc{pdr}} (\vec{c}' \parallel \vec{F}') \), then \( (\vec{c} \parallel \vec{\varphi}) \longrightarrow (\vec{c}' \parallel \vec{\varphi}') \) and \( (\vec{c}' \parallel \vec{\varphi}') \approx (\vec{c}' \parallel \vec{F}') \) for some \((\vec{c}' \parallel \vec{\varphi}')\).
\end{theorem}
\begin{proof}
    This is intuitively clear as each rule has a counterpart, but we need some care to establish the precise correspondence.

    \textsc{(Candidate/sPDR)} can be simulated by \textsc{(Candidate)}.
    To see this, we have to check that \( \models F_n(c) \wedge \neg \alpha(c) \) implies \( \models \varphi_n(c) \wedge \neg \alpha(c) \).
    This follows from the assumption \( F_n(x) \Leftrightarrow \bigvee_{0 \le i \le n} \varphi_i(x) \) and \( \varphi_i(x) \models \alpha(x) \) for \( i < n \).
    The conserve is trivial because \( \varphi_n(x) \models F_n(x) \).

    The correspondence of \textsc{(Decide/sPDR)} and \textsc{(Decide)} can be show similarly.

    It is easy to see that \textsc{(Conflict)} with interpolation \( \psi \) coincides with \textsc{(Conflict/sPDR)} with the same interpolation \( \psi \).
\end{proof}

\begin{remark}\label{rem:pdr:practical}
    sPDR defined by \autoref{fig:PDR} differs from that by Hoder and Bj{\o}rner~\cite{Hoder2012},
    which has an additional rule \textsc{(Induction/PDR)} and a slightly weaker condition for the interpolants as in \textsc{(Conflict/PDR)} below:
    \begin{gather*}
        \infer[\textsc{(Induction/PDR)}]{
            (c_k \dots c_n \parallel F_0 \dots F_n) \longrightarrow_{\textsc{pdr}} (c_k \dots c_n \parallel (F_0 \wedge \psi) \dots (F_{\ell} \wedge \psi) F_{\ell+1} \dots F_n)
        }{
            \iota(y) \vee \big(\exists x. (F_{\ell-1}(x) \wedge \psi(x)) \wedge \tau(x,y))\big) \models \psi(y)
        }
        \\[5pt]
        \infer[\textsc{(Conflict/PDR)}]{
            (c_k \dots c_n \parallel F_0 \dots F_n)
            \longrightarrow_{\textsc{pdr}}
            (c_{k+1} \dots c_n \parallel (F_0 \wedge \psi) \dots (F_{k} \wedge \psi) F_{k+1} \dots F_n)    
        }{
            \iota(y) \vee \big(\exists x. F_{k-1}(x) \wedge \tau(x,y))\big) \models \psi(y)
            \qquad
            \psi(x) \models x \neq c_k
        }
    \end{gather*}
    Notice that, in order to keep the monotonicity \( F_i(x) \models F_{i+1}(x) \), we have to conjoin \( \psi \) to \( F_j \) (\( 0 \le j < k \)) as well.
    This version of PDR coincides with an extension of \IndPDR{} by the following rule corresponding to \textsc{(Induction/PDR)}:
    \begin{equation*}
        \infer[\textsc{(Induction)}]{
            (c_k \dots c_n \parallel \varphi_0 \dots \varphi_n) \longrightarrow (c_k \dots c_n \parallel (\varphi_0 \wedge \psi) \dots (\varphi_{\ell} \wedge \psi) \varphi_{\ell+1} \dots \varphi_n)
        }{
            \iota(y) \models \psi(y)
            \qquad
            \exists x. (\varphi_i(x) \wedge \psi(x)) \wedge \tau(x,y) \models \psi(y)
            \mbox{ for \( i = 0,1,\dots,\ell-1 \)}
        }
    \end{equation*}
    See Appendix~\ref{sec:appx:practical-pdr} for the details.
    \qed
\end{remark}


\begin{remark}\label{rem:pdr:connection-to-impact}\label{rem:originality-of-maximal-conservativity}
    A similarity between PDR and McMillan's lazy abstraction with interpolants was already recognized as in \cite{Een2011,Cimatti2012}.
    They pointed out that a multi-step transition \( (\epsilon \parallel \varphi_0 \dots \varphi_{n-1} \top) \longrightarrow^* (\epsilon \parallel \varphi'_0 \dots \varphi'_{n-1} \varphi'_{n}) \) (with \( \varphi'_n(x) \models \alpha(x) \)) is indeed a calculation of an interpolant.
    Hence PDR ``can be seen as an instance of the lazy abstraction with interpolants algorithm of \cite{McMillan2006}, in which however interpolants are constructed in a very different way''~\cite{Cimatti2012}.
    Here we further investigate the connection, giving a complete characterization of interpolants given by PDR (Theorem~\ref{thm:PDR-MC}).
    To the best of our knowledge, this is the first declarative characterization of interpolants computed by PDR.

    The maximal conservativity heuristic and a na\"ive algorithm for computing a maximally conservative interpolant (Algorithm~\ref{alg:simple-mcr}) can be found in \cite{Vizel2014} (called min-suffix heuristic), which aims to combine the interpolation-based approach and PDR.
    This paper differs from the work in the following points.
    First this paper formally proves that \impact{} with the maximal conservativity heuristic is essentially the same as PDR, whereas \cite{Vizel2014} regards it as a heuristic for interpolation-based model-checkers and does not establish its formal connection to PDR.
    Second, this paper also analyzes other algorithms from the view point of maximal conservativity (see Section~\ref{sec:game-solving}).
    %
    \qed
\end{remark}

\input{pdrmbp.tex}

%% file: pdrmbp.tex
\subsection{Application: Making PDR refutationally complete by model-based projection}
In contrast to the na\"ive algorithm \NaiveMCR{} (Algorithm~\ref{alg:simple-mcr}), which always terminates,
\IndPDR{} (Algorithm~\ref{alg:simple-pdr}) may diverge.
Therefore \MCImpact{} with \IndPDR{} does not have refutational completeness because,
even if an input program (or an input linear CHC system) has an error path of length \( n \), the backend solver \IndPDR{} may diverge during the construction of a proof of the safety of the program within \( (n-1) \) steps.

This subsection develops a variant of PDR with refutational completeness, in order to demonstrate the usefulness of our logical analysis of PDR.
To the best of our knowledge, the algorithm developed here is the first PDR for software model-checking with refutational completeness.\footnote{
    An expert of software model-checking may point out that \recmc{}~\cite{Komuravelli2014,Komuravelli2016}, the algorithm of \spacer{}, is a variant of PDR with refutational completeness.
    Although a proof sketch of refutational completeness of \recmc{} is given in \cite{Komuravelli2016},
    we found a counterexample, which shall be discussed at the end of this subsection.
}

\IndPDR{} may diverge because the loop starting from line 4 in Algorithm~\ref{alg:simple-pdr} may continue indefinitely.
Let
\begin{equation*}
    \mathit{BadReachable}_n ~~:=~~ \{~ c \mid {}\models \exists x. \varphi_{n-1}(x) \wedge \tau(x,c) \wedge \neg \alpha(c) ~\}
\end{equation*}
be the set of bad states that is reachable from \( \varphi_{n-1} \) in one step.
In each loop, \IndPDR{} chooses a state \( M(y) = c \in \mathit{BadRechable}_n \) and tries to refine the immediate subproof \( \varphi_0,\dots,\varphi_{n-1} \) to ensure that \( c \) is in fact unreachable.
\( \mathit{BadReachable}_n \) decreases in each iteration, and the loop terminates when \( \mathit{BadReachable}_n \) becomes empty.
In the worst case, each loop removes only one element, namely the chosen element \( c \), from \( \mathit{BadReachable}_n \).
Hence, if \( \mathit{BadReachable}_n \) is infinite, infinite iteration of the loop may be required to make \( \mathit{BadReachable}_n \) empty.
This divergent behavior has been observed in \cite{Komuravelli2014} for a variant of PDR known as \emph{GPDR}~\cite{Hoder2012}.

From a logical point of view, this phenomenon can be explained as follows.
Recall that the aim of \IndPDR{} is to remove the quantification from \( \exists y. \tau(x,y) \wedge \neg \alpha(y) \) in \IndMCR{} (Algorithm~\ref{alg:rec-mcr}).
The idea is to replace the variable \( y \) with a concrete value, and the point is the logical equivalence
\begin{equation}
    \exists y. \tau(x,y) \wedge \neg \alpha(y)
    \qquad\Leftrightarrow\qquad
    \bigvee_{c} \big(\tau(x,c) \wedge \neg \alpha(c)\big)
    \label{eq:pdr:infinite-decomposition}
\end{equation}
where the right-hand-side is the infinite disjunction.
This allows us to decompose a quantified formula \( \exists y. \tau(x,y) \wedge \neg \alpha(y) \) into an infinite collection of quantifier-free formulas \(\tau(x,c) \wedge \neg \alpha(c)\).
So the essence of \IndPDR{} is a decomposition of a single hard query into infinitely many tractable queries, and the infinity causes divergence in the worst case.

\emph{Model-based projection}~\cite{Komuravelli2014} was introduced to overcome the problem.
\begin{definition}[Model-based projection \cite{Komuravelli2014}]
    Let \( \varphi(\vec{x},\vec{y}) \) be a quantifier-free formula.
    A function \( \mathit{MBP}(\varphi, \{\vec{y}\}, {-}) \) from models \( M \) of \( \varphi \) (i.e.~assignments for \( \vec{x} \) and \( \vec{y} \) such that \( M \models \varphi \)) to quantifier-free formulas \( \psi_M = \mathit{MBP}(\varphi, \{\vec{y}\}, M) \) with \( \mathrm{fv}(\psi_M) \subseteq \{\vec{x}\} \) is a \emph{model-based projection for \(\varphi\) with respect to \(\vec{y}\)} if it satisfies
    \begin{equation}
        \exists \vec{y}. \varphi(\vec{x}, \vec{y}) ~~\Leftrightarrow~~ \bigvee_{M \models \varphi}  \psi_M(\vec{x})
        \qquad\mbox{and}\qquad
        M \models \psi_M(\vec{x})
        \mbox{ for every model \(M\)}
        \label{eq:pdr:mbp-def}
    \end{equation}
    and furthermore the image \( \{ \psi_M \mid M \models \varphi \} \) is finite.
    We assume that a model-based projection exists for every pair of \( \varphi \) and \( \vec{y} \) in the constraint language.
    %
    \qed
\end{definition}
The condition (\ref{eq:pdr:mbp-def}) is an abstraction of the property of the mapping \( \mathit{Subst}(\varphi(\vec{x},\vec{y}), \{\vec{y}\}, M) := \varphi(\vec{x}, M(\vec{y})) \), so model-based projection is a generalization of the decomponition in \autoref{eq:pdr:infinite-decomposition}.
A significant difference is the finiteness of the image: the disjunction \( \bigvee_{M \models \varphi}  \psi_M(\vec{x}) \) is actually a finite disjunction.

By using our inductive description of PDR, it is surprisingly easy to incorporate model-based projection into \IndPDR{}: we simply replace the substitution \( \tau(x,M(y)) \wedge \neg \alpha(M(y)) \) of the value \( M(y) \) for \( y \) (line 6 in Algorithm~\ref{alg:simple-pdr}) with model-based projection \( \mathit{MBP}(\tau(x,y) \wedge \neg \alpha(y), \{y\}, M) \).
The resulting algorithm, which we call \IndPDRmbp{}, is shown in Algorithm~\ref{alg:simple-pdr-mbp}.

\begin{algorithm}[t]
    \caption{~~Inductive PDR with model-based projection}\label{alg:simple-pdr-mbp}
    \begin{algorithmic}[1]
    \Require{Predicates \( \iota(x), \tau(x,y), \alpha(x) \) defining linear CHCs}    
    \Require{A valid partial proof $ \varphi_0,\dots,\varphi_n $ satisfying (\ref{eq:pdr:input-condition})}
    \Ensure{Maximally conservative refinement $ \varphi_0',\dots,\varphi_n'$} \Comment{\textsc{RecPDR-MBP} may fail}
    \Function{\IndPDRmbp}{$\iota,\tau,\alpha;~~ \varphi_0,\dots,\varphi_n$}
    \State {\textbf{if} \( \models \exists x. \iota(x) \wedge \neg \alpha(x) \) \textbf{then} \textbf{fail}}
    \State {\textbf{if} \( \varphi_n(x) \models \alpha(x) \) \textbf{then} \textbf{return} $(\varphi_0,\dots,\varphi_n)$}
    \While{$ \models \exists x. \exists y. \varphi_{n-1}(x) \wedge \tau(x,y) \wedge \neg \alpha(y)$}
        \State \textbf{let} $M$ \textbf{be} an assignment such that $ M \models \varphi_{n-1}(x) \wedge \tau(x,y) \wedge \neg \alpha(y) $
        \State $ \gamma(x) := \mathit{MBP}(\tau(x,y) \wedge \neg \alpha(y), \{y\}, M) $
        \State $ (\varphi_0,\dots,\varphi_{n-1}) := \IndPDRmbp{}(\iota,\tau,\neg\gamma;~~ \varphi_0,\dots,\varphi_{n-1})$
    \EndWhile
    \State{} \Comment{\( \exists x. \varphi_{n-1}(x) \wedge \tau(x,y) \models \varphi_n(y) \wedge \alpha(y) \) holds}
    \State $\varphi_n(y) := \mathit{interpolation}\big(\varphi_{n-1}(x) \wedge \tau(x,y),~~~~ \varphi_n(y) \wedge \alpha(y) \big)$
    \Comment{Nondeterministic}
    \State \textbf{return} $\vec{\varphi}$
    \EndFunction
    \end{algorithmic}
\end{algorithm}

\begin{theorem}\label{thm:pdr-mbp-termination}
    \IndPDRmbp{} always terminates.
\end{theorem}
\begin{proof}
    We prove the claim by induction on \( n \).
    If \( n = 0 \), then \( \varphi_0 = \varphi_n = \iota \) by (\ref{eq:pdr:input-condition}).
    Since
    \begin{equation*}
        \neg (\exists x. \iota(x) \wedge \neg \alpha(x))
        \quad\mbox{iff}\quad
        \forall x. \iota(x) \Rightarrow \alpha(x)
        \quad\mbox{iff}\quad
        \forall x. \varphi_n(x) \Rightarrow \alpha(x),
    \end{equation*}
    \textsc{RecPDR-MBP} fails in line 2 or returns a value in line 3.

    Assume that \( n > 0 \).
    By the induction hypothesis, the recursive call in line 7 terminates.
    It suffices to show the termination of the loop starting from line 4.

    Assume for contradiction that the loop continues indefinitely.
    For each \( i = 1,2,\dots \), let \( M_i \) and \( \gamma_i \) be \( M \) and \( \alpha' \) in the \(i\)-th loop and \( \psi_i \) be \( \varphi_{n-1} \) at the end of the \( i \)-th loop.
    By construction, we have
    \begin{equation*}
        M_i \models \psi_{i-1}(x) \wedge \tau(x,y) \wedge \neg \alpha(x)
        \quad\mbox{and}\quad
        M_i \models \gamma_i(x)
    \end{equation*}
    for each \( i = 1,2,\dots \).
    By the definition of refinement,
    \begin{equation*}
        \psi_{i}(x) \models \neg\gamma_i(x)
        \quad\mbox{and}\quad
        \psi_{i}(x) \models \psi_{i-1}(x)
    \end{equation*}
    where \( \psi_0 \) is \( \varphi_{n-1} \) before the loop.
    Suppose \( i < j \).
    Since \( M_j \models \psi_{j-1}(x) \) and \( \forall k. \big[\psi_{k}(x) \models \psi_{k-1}(x)\big] \), we have \( M_j \models \psi_i(x) \).
    Hence \( M_j \models \neg \gamma_i(x) \) since \( \psi_i(x) \models \neg \gamma_i(x) \).
    As \( M_i \models \gamma_i(x) \), we have \( \gamma_i(x) \neq \gamma_j(x) \) if \( i < j \).
    This contradicts the finiteness assumption of MBP since \( \gamma_i(x) = \mathit{MBP}(\tau(x,y)\wedge\neg\alpha(y), \{y\}, M_i) \) and the formula \( \tau(x,y)\wedge\neg\alpha(y) \) is a constant in the loop.
\end{proof}

It is not difficult to give an abstract transition system corresponding to \IndPDRmbp{} (Algorithm~\ref{alg:simple-pdr-mbp}).
It is shown in \autoref{fig:SPDR-MBP}.

\begin{figure}[t]
    \begin{gather*}
        \infer[\textsc{(Candidate/mbp)}]{
            (\epsilon \parallel \varphi_0 \dots \varphi_n) \longrightarrow (\alpha \parallel \varphi_0 \dots \varphi_n)
        }{
        }
        \\[5pt]
        \infer[\textsc{(Decide/mbp)}]{
            (\gamma_k\,\dots\,\gamma_n \parallel \varphi_0 \dots \varphi_n) \longrightarrow (\gamma_{k-1} \gamma_k\,\dots \gamma_n \parallel \varphi_0 \dots \varphi_n)
        }{
            \genfrac{}{}{0pt}{0}{
                \models \neg \exists x. \iota(x) \wedge \gamma_{k}(x)
                \quad
                M \models \varphi_{k-1}(x) \wedge \tau(x,y) \wedge \gamma_k(y)
            }{
                \gamma_{k-1}(x) = \mathit{MBP}(\tau(x,y) \wedge \gamma_k(y), \{y\}, M)
            }
        }
        \\[5pt]
        \infer[\textsc{(Conflict/mbp)}]{
            (\gamma_{k}\,\dots\,\gamma_n \parallel \varphi_0 \dots \varphi_n) \longrightarrow (\gamma_{k+1} \,\dots\, \gamma_n \parallel \varphi_0 \dots \varphi_{k-1} (\varphi_k \wedge \psi) \varphi_{k+1} \dots \varphi_n)
        }{
            \varphi_{k-1}(x) \wedge \tau(x,y) \models \psi(y)
            \qquad
            \psi(x) \models \neg\gamma_k(x)
        }
    \end{gather*}
    \vspace{-20pt}
    \caption{The transition rules for PDR with model-based projection.}
    \label{fig:SPDR-MBP}
\end{figure}

\subsubsection*{Subtlety of model-based projection}
Algorithms using MBP should be carefully designed.
For example, let \IndPDRmbp{}' be the algorithm obtained by replacing line 6 of \IndPDRmbp{} with
\begin{equation*}
    \gamma(x) ~~:=~~
    \mathit{MBP}(\varphi_{n-1}(x) \wedge \tau(x,y) \wedge \neg \alpha(y),~ \{y\},~ M).
\end{equation*}
This small change breaks the termination property of the algorithm.
This is because \( \varphi_{k-1}(x) \) is changed in each iteration of the loop; hence a different iteration calculates model-based projection of different formulas.
In this case, the finiteness property of model-based projection does not help, since the finiteness is ensured only if model-based projection is applied to a single formula.

We give a small toy example to prove that \IndPDRmbp{}' may diverge.
Let us consider the following unsatisfiable linear CHC system
\begin{equation*}
    S
    ~~=~~
    \{~~ x = 0 \Rightarrow P(x), ~~ P(x) \wedge (y = x+1 \vee y = 1-2x) \Rightarrow P(y), ~~ P(x) \Rightarrow x \le 2 ~~\},
\end{equation*}
so \( \iota(x) = (x=0) \), \( \tau(x,y)=(y=x+1\vee y=1-2x) \) and \( \alpha(x)=(x\le 2) \).
Let \( \varphi_0 = \iota \), \( \varphi_1(x) = (x \le 1) \) and \( \varphi_2(x) = \top \).
Then \( (\varphi_0,\varphi_1,\varphi_2) \) has a refinement, e.g.~\( (x=0,x=1,x \le 2) \).
Let \( \mathit{MBP} \) be an arbitrary model-based projection method.
We define \( \mathit{MBP}' \) by
\begin{equation*}
    \mathit{MBP}'(\psi(x,y), \{y\}, M)
    ~~:=~~
    \begin{cases}
        x = m & \mbox{if \( M(x) = m = \max \{ n < 0 \mid {} \models \exists y. \psi(n,y)\} \)}
        \\
        \mathit{MBP}(\psi(x,y), \{y\},M) & \mbox{otherwise.}
    \end{cases}
\end{equation*}
We also assume that
\begin{equation*}
    \mathit{interpolation}(\psi(x,y), \vartheta(y)) ~~:=~~ \vartheta(y).
\end{equation*}
Then \( \IndPDRmbp{}'(\iota,\tau,\alpha;~~ \varphi_0,\varphi_1,\varphi_2) \) diverges, provided that the model \( M_i \) in the \( i \)-th iteration, \( i = 1,2,\dots \), of the loop is \( (M_i(x), M_i(y)) = (-i, 1+2i) \).
Then \( \varphi_1 \) at the end of \( i \)-th iteration is \( (x \le 1 \wedge x \neq -1 \wedge \dots \wedge x \neq -i) \), and the loop never terminates.

\tkchanged{\recmc{}~\cite{Komuravelli2014,Komuravelli2016}, which is the algorithm used in the state-of-the-art software model-checker \spacer{}}, has the same problem as \IndPDRmbp{}'.
It applies model-based projection to a formula containing the current over- and/or under-approximations, which vary over time; see \textsc{(Query)} rule in \cite[Fig.~7]{Komuravelli2016} (where the existential quantifier in the definition of \( \psi \) shall be replaced with mobel-based projection in a later section).
This observation applies to other formalizations of \spacer{}, e.g.~\cite{Komuravelli2015} (that actually behaves slightly differently).

We give a concrete counterexample.
Let \( S' \) be a non-linear CHC system, given by
\begin{equation*}
    \{~~ 
        \iota(x) \Rightarrow Q(0,x),\quad
        \top \Rightarrow Q(1,x),\quad
        Q(0,x) \wedge Q(1,x) \wedge \tau(x,y) \Rightarrow Q(0,y),\quad
        Q(0,x) \Rightarrow \alpha(x)
    ~~\}.
\end{equation*}
Since \( Q(1,x) = \top \), one can remove \( Q(1,x) \) from \(S'\) by substituting \( \top \) for it.
The resulting constraint set is essentially the same as \( S \), obtained by replacing \( P(x) \) in \( S \) with \( Q(0,x) \).
A badly-behaved model-based projection \( \mathit{MBP}'' \) is similar to \( \mathit{MBP}' \): it basically behaves as \( \mathit{MBP} \) but returns \( x = m \) if \( M(x) = m \) is either the maximum or minimum in \( \{ n < 0 \mid {}\models \exists y. \psi(n,y) \} \).

This divergent behavior is notable because \tkchanged{a journal paper~\cite{Komuravelli2016} proves that \recmc{} always terminates without assuming any assumptions on the model-based projection process and the interpolating theorem prover}.
This shows the nontriviality and subtlety of termination proofs of algorithms with model-based projection.


%% file: game.tex
\section{Maximal Conservativity for Infinite Game Solving}\label{sec:game-solving}
The characterization of PDR by maximal conservativity reveals an unexpected connection between PDR and an efficient procedure for game solving~\cite{Farzan2018}.

Farzan and Kincaid~\cite{Farzan2018} proposed a procedure for game solving over infinite graphs of which rules are defined by linear arithmetic.
Assume that a state is represented by a tuple \( \vec{s} = (s_1, \dots, s_n) \) of data of the logic, as well as actions \( \vec{a}_0 \) and \( \vec{a}_1 \) of Players 0 and 1, respectively.
Players alternatively choose actions, starting from Player 0.
The transition of the game is a partial function represented by \( \tau  \): \( \tau(\vec{s}, \vec{a}_0, \vec{a}_1, \vec{t}) \) is true if and only if \( \vec{t} \) is the state after playing \( \vec{a}_0 \) and then \( \vec{a}_1 \).
Player 0 wins the game if the game does not get stuck.
We are interested in whether Player 0 wins the game.

Their approach is fairly naturally understood as a cyclic proof search, although their paper~\cite{Farzan2018} does not mention the connection.
The winning region of Player 0 is characterized by the greatest solution of
\begin{equation}
    W(\vec{s})
    ~~\Leftrightarrow~~
    \exists \vec{a}_0. \forall \vec{a}_1. \exists \vec{t}. \tau(\vec{s}, \vec{a}_0, \vec{a}_1, \vec{t}) \wedge W(\vec{t}).
    \label{eq:winning-region}
\end{equation}
The definition says that \( \vec{s} \) is winning for Player 0 if there exists an appropriate action \( \vec{a}_0 \) of Player 0 such that, for every action \( \vec{a}_1 \) of Player 1, the next state \( \vec{t} \) exists and it is a winning position of Player 0.
So the game solving is reduced to the validity of \( \iota(\vec{s}) \vdash (\nu W)(\vec{s}) \).

Their procedure can be described as a proof search strategy for sequents of the form \( \varphi(\vec{s}) \vdash (\nu W)(\vec{s}) \).
Application of \textsc{($\nu$-R)} yields a partial proof with an open sequent \( \varphi(\vec{s}) \vdash \delta[\nu W](\vec{s}) \), where \( \delta \) is the right-hand-side of \autoref{eq:winning-region}.
A significant difference from software model-checking is found here: since \( \delta \) contains quantifiers, it is not easy to transform this sequent into \( \psi(\vec{t}) \vdash (\nu W)(\vec{t}) \).
Their idea is to use a proof of a fixed-point-free sequent \( \varphi(\vec{s}) \vdash \delta[\top](\vec{s}) \), from which one can easily obtain a partial proof by replacing \( \top \) with \( (\nu W) \).
To (dis)prove the sequent \( \varphi(\vec{s}) \vdash \delta[\top](\vec{s}) \), they used an external solver, which itself is developed in the first half of their paper~\cite{Farzan2018} using \cite{Farzan2016}.

What is notable is a way to fix the partial proof when \( \varphi(\vec{s}) \vdash \delta[\top](\vec{s}) \) is false.
Recall that \( \varphi(\vec{s}) \vdash (\nu W)(\vec{s}) \) is an open sequent of the current partial proof, which is a leaf.
Their procedure tries to fix the partial proof, by refining a subtree containing this open sequent.
Their procedure first seeks the nearest ancestor of the open sequent \( \varphi(\vec{s}) \vdash (\nu W)(\vec{s}) \) that uses \textsc{($\nu$-R)} rule in the current partial proof.
Let \( \varphi'(\vec{s}') \vdash \delta[\nu W](\vec{s}') \) be its premise.
Then the procedure check the validity of \( \varphi'(\vec{s}') \vdash \delta[\delta[\top]](\vec{s}') \), the two-folded expansion of \( \nu W \).
If the sequent is valid, it replaces the subtree above the ancestor by using the proof; if the sequent is invalid, it seeks the next ancestor and tries to fix the subtree above it.
This is the maximally conservative policy, and their procedure for game solving belongs to the PDR family in this sense.

We conclude with two consequences.
First the notion of maximal conservativity is a useful characterization of PDR that significantly extends the scope of PDR.
Second the maximal conservativity would be useful for implementing an efficient prover: at least, it works well for sequents corresponding to game solving to some extent, as confirmed by the experiments in \cite{Farzan2018}. 

%% file: related.tex
\section{Related work}
\label{sec:related}

We have demonstrated that not only well-known software model checking algorithms, including symbolic execution, bounded model checking~\cite{Biere1999}, 
predicate abstraction~\cite{Ball2001}, lazy abstraction~\cite{Henzinger2002,Henzinger2004,McMillan2006}, and PDR~\cite{Hoder2012,Cimatti2012}, but also an efficient game solving algorithm~\cite{Farzan2018} can be seen as cyclic-proof search.  This section discusses other related proof search and software model-checking algorithms.

\subsection{Proof search in cyclic and Martin-L\"of-style proof systems}

%

Proof search in cyclic~\cite{Sprenger2003,Brotherston2011a} and Martin-L\"of-style~\cite{MartinLoef1971} proof systems finds wide applications in safety/liveness program verification
and entailment checking in first-order and separation logics with inductive predicates.

\paragraph{Program verification.} Brotherston et al.~\cite{Brotherston2008,Brotherston2012,Tellez2020} gave cyclic proof systems for Hoare logic (with the separation logic as the underlying assertion language).
They mentioned a certain part of a proof search corresponds to symbolic execution.
However, they did not establish a connection to modern software model-checkers, of which the main challenge is to find an invariant.
\citet{Unno2017b} presented an inductive proof system tailored to CHC solving and applied it to relational verification.
The search algorithms presented in these studies correspond to bounded model-checking with covering, and the cut-rules (i.e., \textsc{(Cons)} rule in \cite{Tellez2020} and \textsc{(Apply$\bot$)}/\textsc{(ApplyP)} rules in \cite{Unno2017b}) are used only when one needs to check whether an open leaf node is covered.
Researchers have extended SMT solvers to efficiently handle recursive functions and applied them to verification of programs that manipulate structured data~\cite{Suter2010,Suter2011a,Qiu2013,Reynolds2015}.
These work however neither show a connection to modern software model checker nor discuss heuristics to find a good cut-formula using software model checking techniques.

\paragraph{Entailment checking.}
Many papers are devoted to give (efficient) automated theorem provers for entailment problems with inductive predicates~\cite{Brotherston2005,Brotherston2011,Chin2012,Iosif2013,Chu2015,Ta2016,Ta2017} (in particular, of separation logic) and to find decidable fragments of the entailment problems~\cite{Berdine2004,Le2017}.  However, to our knowledge, there is no work that shows connection to modern software model checker.
An interesting question is whether these developments and ideas are applicable to software model-checking.  For example, global trace condition for cyclic proofs gives us a more flexible covering criterion than the na\"ive one used in most software model-checkers.
Also, lemma discovery techniques developed for entailment checking~\cite{Ta2016,Enea2015,Ta2017} or general purpose inductive theorem proving~\cite{Bundy2001} can be useful in software model checking.
Conversely, it is interesting to investigate if Craig interpolation and other heuristics for finding a good cut-formula that have been studied in the software model checking community can be effectively applied as lemma discovery heuristics to inductive theorem proving.

\subsection{Software model checking}

\paragraph{Unified framework.}
Beyer et al. have been developing a configurable software model checker \cpa{}~\cite{Beyer2007b,Beyer2008a,Beyer2011}, which implements various software model checking algorithms~\cite{Beyer2012,Beyer2018,Beyer2020} in a unified framework in a configurable manner.  Though our logical foundation also provides a unified framework, the level of abstraction is significantly different: we reconstruct well-known algorithms from a few simple and \emph{declarative principles}, whereas \cpa{} achieves the reconstruction via \emph{modular design and implementation} of software model checking components.

\paragraph{Constraint logic programming.}
Constraint logic programming (CLP) has been used as a logical foundation of software model checking~\cite{Flanagan2004,Podelski2007a,Bjorner2015a}.  They use CLP to model the target specification as well as the concrete or abstract execution semantics of the target program and then invoke a CLP solver to perform an actual proof-search such as symbolic execution, predicate abstraction, and abstraction refinement.  By contrast, the logical foundation proposed in this paper aims at analyzing and comparing existing and unseen software model checking algorithms as different proof-search strategies in the same cyclic proof system.  In other words, we provide a logical foundation for not only \emph{modeling} but also \emph{solving} safety verification problems.


\paragraph{Interpolating theorem provers.}
We presented declarative characterizations of interpolation-based frame refinement, where we regarded interpolating theorem provers as a black box.  Their internal states and syntactic features (e.g., size and shape) of their outputs cannot be captured by the current cyclic-proof framework.  We plan to investigate the following research questions: Is it possible to see the internal states as partial proofs?  Is it possible to give a declarative characterization for practically ``good'' interpolants: those returned by PDR
and beautiful interpolants~\cite{Albarghouthi2013}.
%

\paragraph{Other approaches to invariant synthesis.}
Other approaches include constraint-based one~\cite{Colon2003,Sankaranarayanan2004a} that uses \(\exists\forall\)-formula solvers and CEGIS~\cite{Solar-Lezama2006} that uses inductive synthesizers~\cite{Sharma2013,Garg2016}.
%
The cyclic-proof framework needs further extensions to capture their internal states: for instance, for the latter we need to annotate each judgement in a partial proof with examples that are iteratively collected by the CEGIS procedure.

%% file: conc.tex
\section{Conclusion and Future Work}
\label{sec:conc}

This paper establishes a tight connection between software model-checkers and cyclic-proof search strategies.
It is worth emphasizing that our connection relates internal states of modern, sophisticated software model-checkers and partial proofs in the cyclic proof system.
The logical viewpoint is useful for understanding model-checking procedures.
A rich structure of (or strong constraints on) partial proofs
allows us to reconstruct important model-checking procedures from a few simple principles: for example, PDR is a proof seach strategy following Principles~\ref{principle:shape-of-goals}, \ref{principle:optimistic-cut}, \ref{principle:refinement} and \ref{principle:maximally-conservative}.
Furthermore the logical characterization of PDR significantly extends the scope of application, as discussed in \autoref{sec:game-solving}.

The most important future work is an empirical study.
An interesting question is about the efficiency of an alternative implementation of PDR, in which maximally conservative interpolants are computed in a different way.
If it is comparable to the original PDR, the development of an efficient maximally-conservatively interpolating theorem prover would be of practical interest.

Future work in another direction is a study of model-checking procedures that are not covered by this paper, such as GPDR~\cite{Hoder2012} and \textsc{Spacer}~\cite{Komuravelli2014} for safety verification of programs with first-order functions.
This class of the verification problem is out of the scope of our framework: although it has an inductive characterization, it does not have the dual, coinductive characterization.

We are also interested in the development of new model-checkers using logical ideas.
For example, the safety verification of programs with functions has an inductive characterization \( (\mu R)(x) \vdash \alpha(x) \),
whereas it has no coinductive characterisation.
The logical analysis of the paper suggests that backward execution looks more natural (cf.~Section~\ref{sec:backward-symbolic-execution}), but empirical evaluation of this suggestion is left for future work.

%% file: cyclic-proof-system.tex
\section{Cyclic Proof System}
This section gives the cyclic proof systems used in this paper, of which details are omitted.
complete list of rules of the cyclic proof system used in this paper.
The description of this section is condensed; see also \cite{Brotherston2011} for the basic ideas of cyclic proof systems.

A \emph{first-order signature} is a pair of sets of function symbols and predicate symbols.
Each symbol is associated with a natural number, called its \emph{arity}.\footnote{So only single-sorted signatures are considered here.}
Assume a first-order signature, fixed in the sequel.

\subsection{Formulas and judgements}
For each natural number \( k \), we assume a (finite or infinite) set \( \mathcal{V}_k \) of \emph{predicate variables} of arity \( k \).
We assume \( \mathcal{V}_k \cap \mathcal{V}_\ell = \emptyset \) if \( k \neq \ell \).

The syntax of \emph{terms} is standard:
\begin{equation*}
    t ::= x \mid f(t_1,\dots,t_k)
\end{equation*}
where \( x \) is a (term or object) variable and \( f \) is an arity-\( k \) function symbol in the signature.
The syntax of \emph{formulas} is also standard:
\begin{equation*}
    \varphi,\psi ::= p(t_1,\dots,t_k) \mid t=t' \mid P(t_1,\dots,t_n) \mid \neg \varphi \mid \varphi \wedge \psi \mid \varphi \vee \psi \mid \forall x. \varphi \mid \exists x. \varphi
\end{equation*}
where \( p \) is an arity-\( k \) predicate symbol in the signature, \( P \) is an arity-\( k \) predicate variable and \( t_1,\dots,t_k \) are terms.

Each predicate variable \( P \in \mathcal{V}_k \) is associated with its defining equation:
\begin{equation*}
    P(x_1,\dots,x_k)
    \quad\Leftrightarrow\quad
    \delta_P(x_1,\dots,x_k).
\end{equation*}
Note that \( \delta_P(x_1,\dots,x_k) \) may contain \( P \) and/or other predicate variables;
every occurrence of a predicate variable in \( \delta_P \) must be positive (i.e.~every occurrence of a predicate variable is under an even number of negation operators).
The meaning of this equation depends on the setting; Section~\ref{sec:logic:lfp} regards \( P \) as the least fixed-point of the equation and Section~\ref{sec:logic:gfp} as the greatest fixed-point.\footnote{Note that \( \delta_P \) may contain a predicate variable other than \( P \).  So the meaning of all predicate variables are defined by mutual recursion, as a solution of the large system of equations.}
We write \( \delta_P[\psi] \) for the formula obtained by replacing \( P \) in \( \delta_P \) with \( \psi \).

We define a \emph{judgement} \( \Gamma \vdash \Delta \) as a pair of finite sequences of formulas.

A \emph{substitution} \( \theta \) is a list \( [t_1/x_1,\dots,t_k/x_k] \) of pairs of terms and variables such that \( x_1,\dots,x_k \) are pairwise distinct.
We write \( \varphi\theta \) for the formula obtained by replacing \( x_i \) in \( \varphi \) with \( t_i \) (for each \( i \)).
A substitution is also applicable to sequences of formulas: if \( \Gamma = \varphi_1,\dots,\varphi_\ell \), then \( \Gamma\theta \) is defined as \( \varphi_1\theta,\dots,\varphi_\ell\theta \).

\subsection{Standard rules for predicate logic}
Figure~\ref{fig:rules-for-predicate-logic} is the list of rules for the sequent calculus for predicate logic (except for the structural rules, which we omit here).
The rules ($\forall$-R) and ($\exists$-L) have the standard side condition: the variable \( x \) has no free occurrence in \( \Gamma \) nor \( \Delta \).
\begin{figure}[t]
    \begin{minipage}{.45\linewidth}
        \leavevmode
    \infrule[Ax]{
        \mathstrut
    }{
        \Gamma, \varphi \vdash \Delta, \varphi
    }
    \infrule[Cut]{
        \Gamma \vdash \Delta, \varphi
        \qquad
        \Gamma, \varphi \vdash \Delta
    }{
        \Gamma \vdash \Delta
    }
    \infrule[Subst]{
        \Gamma \vdash \Delta
    }{
        \Gamma\theta \vdash \Delta\theta
    }
    \infrule[$\top$-L]{
        \Gamma \vdash \Delta
    }{
        \Gamma, \top \vdash \Delta
    }
    \infrule[$\top$-R]{
        \mathstrut
    }{
        \Gamma \vdash \Delta, \top
    }
    \infrule[$\bot$-L]{
        \mathstrut
    }{
        \Gamma, \bot \vdash \Delta
    }
    \infrule[$\bot$-R]{
        \Gamma \vdash \Delta
    }{
        \Gamma \vdash \Delta, \bot
    }
    \infrule[$\wedge$-L]{
        \Gamma, \varphi, \psi \vdash \Delta
    }{
        \Gamma, \varphi \wedge \psi \vdash \Delta
    }
    \infrule[$\wedge$-R]{
        \Gamma \vdash \Delta, \varphi
        \qquad
        \Gamma \vdash \Delta, \psi
    }{
        \Gamma \vdash \Delta, \varphi \wedge \psi
    }
    \end{minipage}
    \begin{minipage}{0.45\linewidth}
        \leavevmode
    \infrule[$\vee$-L]{
        \Gamma, \varphi \vdash \Delta
        \qquad
        \Gamma, \psi \vdash \Delta
    }{
        \Gamma, \varphi \vee \psi \vdash \Delta
    }
    \infrule[$\vee$-R]{
        \Gamma \vdash \Delta, \varphi, \psi
    }{
        \Gamma \vdash \Delta, \varphi \vee \psi
    }
    \infrule[$\neg$-L]{
        \Gamma \vdash \Delta, \varphi
    }{
        \Gamma, \neg \varphi \vdash \Delta
    }
    \infrule[$\neg$-R]{
        \Gamma, \varphi \vdash \Delta
    }{
        \Gamma \vdash \Delta, \neg \varphi
    }
    \infrule[$=$-L]{
        \Gamma[x/y] \vdash \Delta[x/y]
    }{
        \Gamma, x=y \vdash \Delta
    }
    \infrule[$=$-R]{
        \mathstrut
    }{
        \Gamma \vdash \Delta, t=t
    }
    \infrule[$\forall$-L]{
        \Gamma, \varphi[t/x] \vdash \Delta
    }{
        \Gamma, \forall x. \varphi \vdash \Delta
    }
    \infrule[$\forall$-R]{
        \Gamma \vdash \Delta, \varphi
    }{
        \Gamma \vdash \Delta, \forall x. \varphi
    }
    \infrule[$\exists$-L]{
        \Gamma, \varphi \vdash \Delta
    }{
        \Gamma, \exists x. \varphi \vdash \Delta
    }
    \infrule[$\exists$-R]{
        \Gamma \vdash \Delta, \varphi[t/x]
    }{
        \Gamma \vdash \Delta, \exists x. \varphi
    }
    \end{minipage}
    \caption{Standard rules of the sequent calculus for predicate logic}
    \label{fig:rules-for-predicate-logic}
\end{figure}

\subsection{Cyclic proof system for least fixed-points}
\label{sec:logic:lfp}
This subsection describes the cyclic proof system for the predicate logic with least fixed-points.
Hence \( P \) is regarded as the least solution of \( P(x_1,\dots,x_k) \Leftrightarrow \delta_P(x_1,\dots,x_k) \).
To emphasise this interpretation, we write \( P(t_1,\dots,t_n) \) as \( (\mu P)(t_1,\dots,t_n) \).

Figure~\ref{fig:rules-lfp} is the list of rules for the sequent calculus for the predicate logic with least fixed-points.
\begin{figure}[t]
    \begin{minipage}{.45\linewidth}
        \leavevmode
    \infrule[$\mu$-L]{
        \Gamma, \delta_P[\mu P](t_1,\dots,t_k) \vdash \Delta
    }{
        \Gamma, (\mu P)(t_1,\dots,t_k) \vdash \Delta
    }
    \end{minipage}
    \begin{minipage}{0.45\linewidth}
        \leavevmode
    \infrule[$\mu$-R]{
        \Gamma \vdash \Delta, \delta_P[\mu P](t_1,\dots,t_k)
    }{
        \Gamma \vdash \Delta, (\mu P)(t_1,\dots,t_k)
    }
    \end{minipage}
    \caption{The rules for least fixed-points}
    \label{fig:rules-lfp}
\end{figure}

We consider derivations which may be \emph{partial} in the sense that leaves are not necessarily axioms (i.e.~those provable by rules with no premise).
We call a non-axiom leaf an \emph{open leaf} (or a \emph{bud node}).
A \emph{companion} of an open leaf \( l \) is an internal node \( n \) in the derivation such that \( l \) and \( n \) are labelled by the same sequent.

A \emph{pre-proof} is a (possibly partial) derivation with a function that assigns a companion for each open leaf.
Hence it is a directed and pointed graph (instead of a tree) such that
\begin{itemize}
    \item each node is labelled by a sequent, and
    \item if \( n_1,\dots,n_k \) is the set of nodes immediately reachable from \( n \), there is a proof rule whose conclusion is (the label of) \( n \) and whose premises are (the labels of) \( n_1,\dots,n_k \).
\end{itemize}
Here a graph is \emph{pointed} if it is equipped with a chosen node that we call the \emph{initial node}.

As we have seen, a pre-proof is not necessarily valid.
We introduce a correctness criterion, called the \emph{global trace condition}.
We need some auxiliary notions.

An \emph{infinite path} \( \pi = n_0 n_1 n_2 \dots \) of a pre-proof is an infinite path of the corresponding graph (starting from the initial node \( n_0 \)): the first node \( n_0 \) is the conclusion of the derivation and \( n_{i+1} \) is a premise of \( n_i \).
An \emph{infinite trace} of an infinite path \( \pi = n_0 n_1 \dots \) is an infinite sequence \( \varphi_0 \varphi_1 \dots \) of (occurrences of) formulas in the derivation that satisfies the following conditions:
\begin{itemize}
    \item \( \varphi_i \) appears in the sequent of \( n_i \).
    \item \( \varphi_{i+i} \) is an occurrence of a formula in the premise \( n_{i+1} \) of \( n_i \), an occurrence which corresponds to \( \varphi_i \) in \( n_i \).
    \begin{itemize}
        \item If \( \varphi_i \) is the principal occurrence and the rule concluding \( n_i \) is not \textsc{($\wedge$-L)} nor \textsc{($\vee$-R)}, then \( \varphi_{i+1} \) is the unique formula outside of \( \Gamma,\Delta \).  For example, if \( n = (\Gamma, \psi_1 \vee \psi_2 \vdash \Delta) \), \( \varphi_{i} = \psi_1 \vee \psi_2 \), the rule is
        \[
            \dfrac{
                \Gamma, \psi_1 \vdash \Delta
                \qquad
                \Gamma, \psi_2 \vdash \Delta
            }{
                \Gamma, \psi_1 \vee \psi_2 \vdash \Delta
            },
        \]
        and \( n_{i+1} \) is the left premise, then \( \varphi_{i+1} = \psi_1 \).
        The rules \textsc{($\wedge$-L)} nor \textsc{($\vee$-R)} are exceptional: these rules has two formulas outside of \( \Gamma,\Delta \) and \( \varphi_{i+1} \) is one of them.
        \item If \( \varphi_i \) is not principal, i.e.~it is in \( \Gamma \) or \( \Delta \), then \( \varphi_{i+1} \) is the same formula in \( \Gamma \) or \( \Delta \).
    \end{itemize}
\end{itemize}
Given a trace \( \varphi_0,\varphi_1,\dots \), the \( i \)-th occurrence \( \varphi_i \) is a \emph{progressing point} if \( \varphi_i \) is the principal occurrence of \textsc{($\mu$-L)} rule.
A trace is \emph{infinitely progressing} if it has infinitely many progressing points.

A pre-proof is a \emph{proof} if every infinite path has an infinitely progressing trace.

\subsection{Cyclic proof system for greatest fixed-points}
\label{sec:logic:gfp}
The cyclic proof system for the predicate logic with greatest fixed-points is essentially the same as that for least fixed-points.
Semantically, the difference is that \( P \) is regarded as the greatest solution of \( P(x_1,\dots,x_k) \Leftrightarrow \delta_P(x_1,\dots,x_k) \) (so \( P(t_1,\dots,t_n) \) is written as \( (\nu P)(t_1,\dots,t_n) \)).

The proof rules for greatest fixed-points, which are the same as those for least fixed-points, are given in Figure~\ref{fig:rules-gfp}.
The notions of pre-proofs, paths and traces are the same as in the proof system for least fixed-points.
The difference is that \textsc{(\( \nu \)-R)} is progressive whereas \textsc{(\(\nu\)-L)} is not.
Hence the global trance condition requires that each infinite path has an infinite trace that appears at the principal position of \textsc{(\( \nu \)-R)} rule infinitely many times.
\begin{figure}[t]
    \begin{minipage}{.45\linewidth}
        \leavevmode
    \infrule[$\nu$-L]{
        \Gamma, \delta_P[\nu P](t_1,\dots,t_k) \vdash \Delta
    }{
        \Gamma, (\nu P)(t_1,\dots,t_k) \vdash \Delta
    }
    \end{minipage}
    \begin{minipage}{0.45\linewidth}
        \leavevmode
    \infrule[$\nu$-R]{
        \Gamma \vdash \Delta, \delta_P[\nu P](t_1,\dots,t_k)
    }{
        \Gamma \vdash \Delta, (\nu P)(t_1,\dots,t_k)
    }
    \end{minipage}
    \caption{The rules for greatest fixed-points}
    \label{fig:rules-gfp}
\end{figure}

\begin{remark}
    This paper does not deal with formulas which contains both least and greatest fixed-points.
    The idea of cyclic proofs is applicable to logics with both least and greatest fixed-points, but we do not give a detail of cyclic proof systems for such logics because they need more complicated global trace conditions.
    \qed
\end{remark}

%% file: appx-pdr.tex
\section{Simulating PDR in \cite{Hoder2012}}
\label{sec:appx:practical-pdr}
Recall PDR and an extension of \IndPDR{} discussed in Remark~\ref{rem:pdr:practical}.
Let \( \sim \) be a relation between configurations defined by
\begin{align*}
    \textstyle
    &
    (c_k \dots c_n \parallel \varphi_0 \dots \varphi_n)
    \sim
    (c_k \dots c_n \parallel F_0 \dots F_n)
    \qquad\mbox{if and only if}\qquad
    \mbox{for every \( 0 \le j \le n \),}
    \\
    &
    \textstyle
    \hspace{30pt}
    \bigvee_{0 \le i \le j} \varphi_i(x) \models F_j(x)
    \quad\mbox{and}\quad
    (F_j(x) \wedge \neg \bigvee_{0 \le i \le j} \varphi_i(x)) \quad\mbox{is safe within \( (n-j) \) steps.}
\end{align*}
We show that the relation \( \sim \) is a weak bisimulation.
\begin{theorem}
    Suppose that \( (\vec{c} \parallel \vec{\varphi}) \sim (\vec{c} \parallel \vec{F}) \).
    If \( (\vec{c} \parallel \vec{\varphi}) \longrightarrow (\vec{c}' \parallel \vec{\varphi}') \), then \( (\vec{c} \parallel \vec{F}) \longrightarrow^*_{\textsc{pdr}} (\vec{c}' \parallel \vec{F}') \) and \( (\vec{c}' \parallel \vec{\varphi}') \sim (\vec{c}', \parallel \vec{F}') \) for some \((\vec{c}' \parallel \vec{F}')\).
    Conversely, if \( (\vec{c} \parallel \vec{F}) \longrightarrow_{\textsc{pdr}} (\vec{c}' \parallel \vec{F}') \), then \( (\vec{c} \parallel \vec{\varphi}) \longrightarrow^* (\vec{c}' \parallel \vec{\varphi}') \) and \( (\vec{c}' \parallel \vec{\varphi}') \sim (\vec{c}' \parallel \vec{F}') \) for some \((\vec{c}' \parallel \vec{\varphi}')\).
\end{theorem}
\begin{proof}
    This is intuitively clear as each rule has a counterpart, but we need some care to establish the precise correspondence.
    Note that \( (\vec{c} \parallel \vec{F}) \longrightarrow_{\textsc{pdr}}^* (\vec{c} \parallel \vec{F}') \) with \( F'_j(x) \Leftrightarrow \bigvee_{0 \le i \le j} \varphi_i(x) \) for every \( 0 \le j \le n \) by iteratively applying \textsc{(Induction/pdr)} for \( \ell = 1,2,\dots,n \).
    This is useful when simulating \(\longrightarrow\) by \(\longrightarrow_{\textsc{pdr}}\).

    Correspondence of \textsc{(Induction/PDR)} and \textsc{(Induction)} is easy.

    \textsc{(Candidate/PDR)} can be simulated by \textsc{(Candidate)}.
    To see this, we have to check that \( \models F_n(c) \wedge \neg \alpha(c) \) implies \( \models \varphi_n(c) \wedge \neg \alpha(c) \).
    This follows from the assumption \( F_n(x) \wedge (\bigvee_{0 \le i \le n} \varphi_i(x)) \) is safe within \( 0 \) steps, i.e.~\( F_n(x) \wedge (\bigvee_{0 \le i \le n} \varphi_i(x)) \models \alpha(x) \).
    The converse is trivial because \( \varphi_n(x) \models F_n(x) \).

    The correspondence of \textsc{(Decide/PDR)} and \textsc{(Decide)} can be show similarly.

    \textsc{(Conflict)} can be simulated by \textsc{(Conflict/PDR)} and \textsc{(Induction/PDR)}.
    By using \textsc{(Induction/PDR)}, we can assume without loss of generality that \( F'_j(x) \Leftrightarrow \bigvee_{0 \le i \le j} \varphi_i(x) \) for every \( 0 \le j \le n \).
    Suppose \( \varphi_{k-1}(x) \wedge \tau(x,y) \models \psi(y) \) and \( \psi(x) \models x \neq c_k \).
    Let \( \psi'(x) = F_{k-1}(x) \vee \psi(x) \).
    Then \textsc{(Conflict/PDR)} with the interpolation \( \psi' \) gives a desired transition.

    \textsc{(Conflict/PDR)} can be simulated by \textsc{(Conflict)} and \textsc{(Induction)}.
    Assume \( \iota(y) \vee \big(\exists x. F_{k-1}(x) \wedge \tau(x,y) \big) \models \psi(y) \).
    By \( \varphi_i(x) \models F_{k-1}(x) \) for every \( i \le k-1 \), we have \( \iota(y) \models \psi(y) \) and \( \exists x. \varphi_i(x) \wedge \tau(x,y) \models \psi(y) \) for every \( i \le k-1 \).
    By \textsc{(Conflict)} and \textsc{(Induction)},
    \begin{align*}
        (c_k \dots c_n \parallel \varphi_0 \dots \varphi_n)
        &\qquad\longrightarrow\qquad
        (c_{k+1} \dots c_n \parallel \varphi_0 \dots \varphi_{k-1} (\varphi_k \wedge \psi) \varphi_{k+1} \dots \varphi_n)
        \\
        &\qquad\longrightarrow\qquad
        (c_{k+1} \dots c_n \parallel (\varphi_0 \wedge \psi) \dots (\varphi_k \wedge \psi) \varphi_{k+1} \dots \varphi_n)
    \end{align*}
    and this is the required transition.
\end{proof}

%% file: appx-diverge-spacer.tex
\section{Diverging behaviour of \recmc{}}
We shows that \recmc{}~\cite{Komuravelli2014,Komuravelli2016} is possibly diverging even for bounded depth.

\paragraph{CHC}
For ease of explanation, we use a CHC system that slightly differs from that in the main text.
The CHC system has multiple predicate variables \( P \), \( Q \) and \( R \); 
if you would like to have a counterexample with a single predicate variable (in order to fit a formalisation of, for example, \cite{Komuravelli2015}), you can replace \( P(x) \), \( Q(x) \) and \( R(x) \) with \( S(0,x) \), \( S(1,x) \) and \( S(2,x) \), respectively, for a fresh predicate variable \( S \).

The non-linear CHC system is given by
\begin{equation*}
    \{~~ 
        x = 0 \Rightarrow P(x),\quad
        \top \Rightarrow Q(x),\quad
        P(x) \wedge Q(x) \Rightarrow R(x),\quad
        R(x) \Rightarrow \bot
    ~~\}.
\end{equation*}
This is unsatisfiable, since \( P(0) \), \( Q(0) \) and \( R(0) \) must be true, contradicting \( R(x) \Rightarrow \bot \).

For an arbitrary model-based projection \( \mathit{MBP} \), we define a badly-behaved model-based projection \( \mathit{MBP}' \) by
\begin{equation*}
    \mathit{MBP}'(\psi(x,\vec{y}), \{\vec{y}\}, M)
    ~~:=~~
    \begin{cases}
        x = m & \mbox{if \( M(x) = m = \max \{ n > 0 \mid {} \models \exists \vec{y}. \psi(n,y)\} \)}
        \\
        x = m & \mbox{if \( M(x) = m = \min \{ n > 0 \mid {} \models \exists \vec{y}. \psi(n,y)\} \)}
        \\
        \mathit{MBP}(\psi(x,y), \{y\},M) & \mbox{otherwise.}
    \end{cases}
\end{equation*}
This satisfies the requirements for model-based projection: note that the cardinality of the image of \( \mathit{MBP}'(\psi(x,\vec{y}), \{\vec{y}\}, ({-})) \) is that of \( \mathit{MBP} \) plus \( 2 \) and thus finite.

We also assume that
\begin{equation}
    \mathit{interpolation}(\psi(x,y), \vartheta(y)) ~~:=~~ \vartheta(y).
    \label{eq:appx:interpolation}
\end{equation}

\paragraph{Divergence of\/ \recmc{}}
The behaviour of \recmc{}, specialised to the above setting, can be explained as follows.

We run the procedure to check the safety of the CHC with depth \( 2 \).
\recmc{} in maintains both over- and under-approximations for each predicate for each depth.
We write \( \displaystyle A_i = \left[ \begin{array}{c} \varphi \\ \psi \end{array} \right] \) to mean that the under- and over-approximations of predicate \( A \) at level \( i \) are \( \varphi \) and \( \psi \), respectively.
The expressions \( \overline{A_i} \) and \( \underline{A_i} \) indicate \( \varphi \) and \( \psi \), respectively.

Consider the following situation:
\begin{align*}
    P_1 &= \left[ \begin{array}{c} \top \\ \bot \end{array} \right]
    &
    Q_1 &= \left[ \begin{array}{c} \top \\ \bot \end{array} \right]
    &
    R_1 &= \left[ \begin{array}{c} \bot \\ \bot \end{array} \right]
    \\
    P_2 &= \left[ \begin{array}{c} \top \\ \bot \end{array} \right]
    &
    Q_2 &= \left[ \begin{array}{c} \top \\ \bot \end{array} \right]
    &
    R_2 &= \left[ \begin{array}{c} \top \\ \bot \end{array} \right].
\end{align*}
We would like to construct a refinement such that \( \overline{R}_2(x) \Rightarrow \bot \); this is impossible and \( x=0 \) is the unique counterexample.

\recmc{} maintains a set of \emph{queries}, which are constraints that the process tries to make true.
In the current situation, it starts from the singleton \( \{ R_2(x) \models \bot \} \).\footnote{
    The representation of queries differs from that in \cite{Komuravelli2016}.
    We describe a query by the condition \( A(\vec{x}) \models \varphi(\vec{x}) \) that a predicate should satisfy, but \cite{Komuravelli2016} represent it as a pair \( (A(\vec{x}), \neg \varphi(\vec{x})) \) expressing that \( A \) should not intersect with \( \neg \varphi \) (i.e.~\( \models \neg \exists \vec{x}. A(\vec{x}) \wedge \neg \varphi(\vec{x}) \)).
    The difference is not essential, as it is trivial to move from one to the other by negating the formula.
}

The behaviour of \recmc{} can be understood as rewriting of approximations and queries.
We do not give the complete description of the rules, but readers familiar with \recmc{} should easily understand which rule is rule at each step.

\recmc{} first checks if \( \overline{R_2}(x) \models \bot \) by changing only \( R_2 \).
It is impossible because
\begin{equation*}
        \overline{P_1}(x) \wedge \overline{Q_1}(x) ~~\Rightarrow~~ \overline{R_2}(x)
\end{equation*}
should hold, which requires \( \top \wedge \top \Rightarrow \overline{R_2}(x) \).
The impossibility can be witnessed by the model \( M(x) = 1 \), which satisfies the left-hand-side (i.e.~\( M \models \top \)) and thus \( \overline{R_2}(x) \) but violates the requirement (i.e.~\( M \not\models \bot \)).

So one needs to strengthen \( \overline{P_1} \) and/or \( \overline{Q_1} \).
For example, it may try to strengthen \( \overline{Q_1} \).
In this case, \recmc{} issues a new query \( Q_1(x) \models \psi_1(x) \), where \( \psi_1(x) \) is obtained by the model-based projection:
\begin{equation*}
    \mathit{MBP}'(\overline{P_1}(x) \wedge \neg \bot,~~\{\},~~M).
\end{equation*}
By definition of \( \mathit{MBP}' \), we have \( \psi_1(x) = (x = 1) \).
The new queries requires that \( Q_1(x) \) does not intersect with \( (x=1) \), i.e.~\( Q_1(x) \models x \neq 1 \).

This requirement cannot be fulfilled because of \( \top \Rightarrow Q(x) \).
This impossibility is recorded by updating the underappoximation \( \underline{Q_1} \) of \( Q_1 \), i.e.~\( \top \not\models \psi_1(x) \); it has the unique counter-model \( M \) such that \( M(x) = 1 \).
Update of underapproximation also uses by model-based projection:
\begin{equation*}
    \mathit{MBP}'(\top \wedge \neg \psi_1, \{\}, M),
\end{equation*}
which must be \( (x=1) \) (since \( \top \wedge \neg \psi_1 \) denotes the singleton set \( \{ 1 \} \)).
Now
\begin{equation*}
    Q_1 = \left[\begin{array}{c} \top \\ x = 1 \end{array}\right].
\end{equation*}
We remove the query \( Q_1(x) \models x \neq 1 \), since it has been made clear that it cannot be fulfilled.

We need to reconsider the query \( R_2(x) \models \bot \).
A model \( M \) with \( x = 1 \) is still a counterexample.
As the attempt to strengthen \( Q_1 \) failed, we try to strengthen \( P_1 \).
In this case, \recmc{} issues a new query \( Q_1(x) \models \vartheta_1(x) \), where \( \vartheta_1(x) \) is obtained by the model-based projection:\footnote{
    Here we use the underapproximation of \( Q \) to generate the query.  The general rule is that, from a pair of a constraint
    \begin{equation*}
        A^1(x^1) \wedge \dots \wedge A^n(x^n) ~~\Rightarrow~~ B(y)
    \end{equation*}
    a query \( B_\ell(y) \models \varphi(y) \),
    a new query for \( A^j \) is obtained by model-based projection of
    \begin{equation*}
        \overline{A_{\ell-1}}^1(x^1) \wedge \dots \wedge \overline{A_{\ell-1}}^{j-1} \wedge \underline{A_{\ell-1}}^{j+1} \wedge \dots \wedge \underline{A_{\ell-1}}^n(x^n) \wedge \neg \varphi(y).
    \end{equation*}    
}
\begin{equation*}
    \mathit{MBP}'(\underline{Q_1}(x) \wedge \neg \bot,~~\{\},~~M).
\end{equation*}
Since \( \underline{Q_1}(x) \wedge \neg \bot \) is equivalent to \( x=1 \), we have \( \vartheta_1(x) = (x = 1) \).
The new query is \( P_1(x) \models \neg \vartheta_1(x) \), i.e.~\( P_1(x) \models x \neq 1 \).

One can strengthen \( P_1(x) \) to meet the requirement (since \( x = 0 \Rightarrow P_1(x) \) is the unique constraint for \( P_1 \)).
Then \recmc{} strengthens \( P_1(x) \) by conjoining an interpolation
\begin{equation*}
    \mathit{interpolation}(\bot,~~x\neq 1).
\end{equation*}
Our interpolating theorem prover returns \( x \neq 1 \) as the interpolation (recall the assumption (\ref{eq:appx:interpolation})).
The current approximations are
\begin{align*}
    P_1 &= \left[ \begin{array}{c} x \neq 1 \\ \bot \end{array} \right]
    &
    Q_1 &= \left[ \begin{array}{c} \top \\ x = 1 \end{array} \right]
    &
    R_1 &= \left[ \begin{array}{c} \bot \\ \bot \end{array} \right]
    \\
    P_2 &= \left[ \begin{array}{c} \top \\ \bot \end{array} \right]
    &
    Q_2 &= \left[ \begin{array}{c} \top \\ \bot \end{array} \right]
    &
    R_2 &= \left[ \begin{array}{c} \top \\ \bot \end{array} \right].
\end{align*}

As a result of the update of \( \overline{P_1}(x) \), we can update the \huchanged{over-approximation} of \( R_2 \) by an argument similar to the case of \( P_1 \):
\begin{align*}
    P_1 &= \left[ \begin{array}{c} x \neq 1 \\ \bot \end{array} \right]
    &
    Q_1 &= \left[ \begin{array}{c} \top \\ x = 1 \end{array} \right]
    &
    R_1 &= \left[ \begin{array}{c} \bot \\ \bot \end{array} \right]
    \\
    P_2 &= \left[ \begin{array}{c} \top \\ \bot \end{array} \right]
    &
    Q_2 &= \left[ \begin{array}{c} \top \\ \bot \end{array} \right]
    &
    R_2 &= \left[ \begin{array}{c} x \neq 1 \\ \bot \end{array} \right].
\end{align*}
However \( \overline{R_2}(x) \) still has infinite counterexamples to \( \overline{R_2}(x) \models \bot \).

We can do the same for the model \( M(x) = 2 \).
The new query \( Q_1(x) \models \psi_2(x) \) is obtained by model-based projection:
\begin{equation*}
    \psi_2(x) ~~=~~
    \mathit{MBP}'(\overline{P_1}(x) \wedge \neg \bot,~~\{\},~~M).
\end{equation*}
Since \( 2 \) is the minimum positive value satisfying \( \overline{P_1}(x) \wedge \neg \bot \), we have \( \psi_2(x) = (x = 2) \) by the definition of \( \mathit{MBP}' \).
So model-based projection \( \mathit{MBP}' \) dose not generalise the counterexample in this case.
The new query leads to an update of \( \underline{Q_1} \), resulting in
\begin{align*}
    P_1 &= \left[ \begin{array}{c} x \neq 1 \\ \bot \end{array} \right]
    &
    Q_1 &= \left[ \begin{array}{c} \top \\ x = 1 \vee x = 2 \end{array} \right]
    &
    R_1 &= \left[ \begin{array}{c} \bot \\ \bot \end{array} \right]
    \\
    P_2 &= \left[ \begin{array}{c} \top \\ \bot \end{array} \right]
    &
    Q_2 &= \left[ \begin{array}{c} \top \\ \bot \end{array} \right]
    &
    R_2 &= \left[ \begin{array}{c} x \neq 1 \\ \bot \end{array} \right].
\end{align*}

We can then issue a new query \( P_1(x) \models \neg \vartheta_2(x) \) to \( P_1 \), where
\begin{equation*}
    \vartheta_2(x)
    ~~=~~
    \mathit{MBP}'(\underline{Q_1}(x) \wedge \neg \bot,~~\{\},~~M)
\end{equation*}
with \( M(x) = 2 \).
Note that \( \underline{Q_1}(x) \wedge \neg \bot \) is equivalent to \( 1 \le x \le 2 \).
Since \( 2 \) is the maximum positive integer satisfying \( \underline{Q_1}(x) \wedge \neg \bot \), \( \mathit{MBP}' \) does not generalise the model \( M(x) = 2 \) even in this case.
So \( \vartheta_2(x) = (x = 2) \).
This query causes update of the \huchanged{over-approximation} of \( P_1 \), and then \( R_2 \):
\begin{align*}
    P_1 &= \left[ \begin{array}{c} x \neq 1 \wedge x \neq 2 \\ \bot \end{array} \right]
    &
    Q_1 &= \left[ \begin{array}{c} \top \\ x = 1 \vee x = 2 \end{array} \right]
    &
    R_1 &= \left[ \begin{array}{c} \bot \\ \bot \end{array} \right]
    \\
    P_2 &= \left[ \begin{array}{c} \top \\ \bot \end{array} \right]
    &
    Q_2 &= \left[ \begin{array}{c} \top \\ \bot \end{array} \right]
    &
    R_2 &= \left[ \begin{array}{c} x \neq 1 \wedge x \neq 2 \\ \bot \end{array} \right].
\end{align*}

The same argument for the counterexample \( M(x) = 3 \), we reaches
\begin{align*}
    P_1 &= \left[ \begin{array}{c} x \neq 1 \wedge x \neq 2 \wedge x \neq 3 \\ \bot \end{array} \right]
    &
    Q_1 &= \left[ \begin{array}{c} \top \\ x = 1 \vee x = 2 \vee x = 3 \end{array} \right]
    &
    R_1 &= \left[ \begin{array}{c} \bot \\ \bot \end{array} \right]
    \\
    P_2 &= \left[ \begin{array}{c} \top \\ \bot \end{array} \right]
    &
    Q_2 &= \left[ \begin{array}{c} \top \\ \bot \end{array} \right]
    &
    R_2 &= \left[ \begin{array}{c} x \neq 1 \wedge x \neq 2 \wedge x \neq 3 \\ \bot \end{array} \right]
\end{align*}
and this process can continue indefinitely.